\documentclass{article}

\usepackage{xifthen}
\usepackage{amsfonts}
\usepackage{amsmath}
\usepackage{amssymb}
\usepackage{stmaryrd}
\usepackage{mathtools}
\usepackage{pgfplots}
\pgfplotsset{compat=1.16}
\usepackage{url}
\usepackage{wrapfig}

\usepackage{amsthm}
\newtheorem{theorem}{Theorem}
\newtheorem{lemma}{Lemma}
\newtheorem{proposition}{Proposition}
\newtheorem{corollary}{Corollary}
\theoremstyle{definition}
\newtheorem{definition}{Definition}

\newtheorem{remark}{Remark}

\usepackage{hyperref}
\usepackage{cleveref}
\usepackage[numbers]{natbib}
\usepackage[margin=1in]{geometry}

\usepackage{tikz}
\usetikzlibrary{automata}
\usetikzlibrary{positioning}
\usetikzlibrary{fit}
\usetikzlibrary{shapes}
\usetikzlibrary{arrows.meta}
\usetikzlibrary{calc}
\usetikzlibrary{backgrounds}
\usetikzlibrary{hobby}
\usepackage{tikz-cd}

\tikzset{every loop/.style={min distance=10mm,in=-45,out=45,looseness=10}}

\usepackage{bookmark}

\usepackage{shuffle}
\usepackage{xfp}

\usepackage{authblk}

\newcommand{\nats}{\mathbb{N}}
\newcommand{\reals}{\mathbb{R}}

\def\rmdef{\stackrel{\mbox{\rm {\tiny def}}}{=}} 
\newcommand{\bb}[1]{\mathbb{#1}}

\newcommand{\upto}{\uparrow}
\newcommand{\tn}[1]{\textnormal{#1}}

\newcommand{\seq}[1]{\left\langle #1 \right\rangle}   
\newcommand{\set}[1]{\left\{ #1 \right\}}
\newcommand{\norm}[2]{\ifthenelse{\equal{#2}{}}{||#1||}{||#1||_{#2}}}

\newcommand{\Ptime}{\textnormal{PTIME}}

\newcommand{\UP}{\textnormal{UP}}
\newcommand{\coUP}{\textnormal{coUP}}

\newcommand{\Exp}{\textnormal{EXPTIME}}
\newcommand{\NExp}{\textnormal{NEXPTIME}}
\newcommand{\coNExp}{\textnormal{coNEXPTIME}}

\newcommand{\dist}{\mathsf{dist}}

\newcommand{\Min}{\tn{Min}}   
\newcommand{\Max}{\tn{Max}} 

\newcommand{\FPlays}{\mathsf{fPlays}}   
\newcommand{\Plays}{\mathsf{Plays}} 
\newcommand{\last}{\mathsf{last}} 
  
\newcommand{\eE}{\mathbb E}

\newcommand{\sx}{\chi} 
\newcommand{\sn}{\nu}

\newcommand{\cyl}{\mathsf{cyl}}
\newcommand{\supp}{\mathsf{supp}}

\newcommand{\M}{\mathcal{M}}
                                                                                                                
\newcommand{\VAL}{\mathsf{Val}}
\newcommand{\Val}[1]{\mathsf{Val}\left( #1 \right)}
\newcommand{\UVal}[1]{\overline{\mathsf{Val}} \left( #1 \right)}
\newcommand{\LVal}[1]{\underline{\mathsf{Val}} \left( #1 \right)}

\newcommand{\subseq}[3]{#1|^{#3}_{#2}}

\newcommand{\cla}[1]{{\color{blue}#1}}
\newcommand{\clb}[1]{{\color{red}#1}}
\newcommand{\clc}[1]{{\color{teal}#1}}
\newcommand{\cld}[1]{{\color{violet}#1}}
\newcommand{\cle}[1]{{\color{purple}#1}}

\title{Discounting the Past\thanks{This work is supported by the National Science Foundation grant 2009022 and by a CU Boulder Research and Innovation Office grant.}}

\author{Taylor Dohmen and Ashutosh Trivedi}

\affil{University of Colorado, Boulder, USA}

\begin{document}

\maketitle
\begin{abstract}
    Stochastic games with discounted payoff, introduced by Shapley, model adversarial interactions in stochastic environments where two players try to optimize a discounted sum of rewards.
In this model, long-term weights are geometrically attenuated based on the delay in their occurrence.
We propose a temporally dual notion---called past-discounting---where agents have geometrically decaying memory of the rewards encountered during a play of the game. 
We study objective functions based on past-discounted weight sequences and examine the corresponding stochastic games with liminf, discounted, and mean payoffs.  
For objectives specified as the limit inferior of past-discounted reward sequences, we show that positional determinacy fails and that optimal strategies may require unbounded memory.
To overcome this obstacle, we study an approximate windowed objective based on the idea of using sliding windows of finite length to examine infinite plays.
On the other hand, for objectives specified as the discounted and average limits of past-discounted reward sequences we establish determinacy in mixed stationary strategies in the setting of concurrent stochastic games and show how the values of these games may be computed via reductions to standard discounted and mean-payoff games.

\end{abstract}

\section{Introduction}
\label{sec:intro}
\emph{Time preference} refers to the tendency of rational agents to place a higher value on the \emph{desirable outcomes} received at an earlier time as compared to a later time. 
This phenomenon is often codified as discounted payoff in mathematical models from diverse disciplines such as economics and game theory, control theory, and reinforcement learning.
It stands to reason, then, that for \emph{undesirable outcomes} the time preference of rational agents may be dual: a tendency to regret less the undesirable outcomes experienced in the distant past compared to those encountered recently. 
We study a temporally dual notion of ``future'' discounting as \emph{past-discounting} and optimization and games over such preferences.

Two-player zero-sum stochastic games provide a natural model for adversarial interactions between rational agents with competing objectives in uncertain settings. 
Beginning with an initial configuration of an arena, these games proceed in discrete time, and at each step the players---named Min and Max---concurrently choose (stochastically) from a set of state-dependent actions. 
Based on their choices, a scalar weight is determined with the interpretation that this quantity represents a ``payment'' from Min to Max. 
Given the current state and the players' action pair, a probabilistic transition
function determines the next state and the process repeats infinitely. 
The goal of Max is to maximize a given objective---typically specified as some functional into the real numbers---defined on infinite sequences of payments $x = \seq{x_n}_{n\geq0}$, while the goal of Min is the opposite. 
The most popular objectives are the limit payoff, the discounted payoffs, which are well-defined for discount factors $\lambda \in [0, 1)$, and the mean payoff:
\begin{align}
	L(x) &\rmdef \liminf_{n \to \infty} x_n \tag{Limit Payoff}\\
	D_\lambda(x) &\rmdef \lim_{n \to \infty} \sum^n_{k=0} \lambda^k x_k, \tag{Discounted Payoff} \\
	M(x) &\rmdef \liminf_{n \to \infty} \frac{1}{n+1} \sum^n_{k=0} x_k. \tag{Mean Payoff}
\end{align}

The discount factor $\lambda$ can be interpreted as the complement of a probability $(1 - \lambda)$ that the game will halt at any given point in time.
At time $k$, the quantity transferred to \Max{} from \Min{} is scaled by the probability that the game continues: $\lambda^k$.
The players are interested in optimizing, in expectation, the total accumulated reward before termination.
For an infinite sequence $x = \seq{x_n}_{n \geq 0}$, the discounted sum  $(1 - \lambda) D_\lambda(x)$ may be viewed as a weighted average---a convex combination, in fact---in which
elements occurring earlier in the sequence are weighted more heavily.  
On the other hand, when agents prefer long-run rewards over initial rewards, a
more traditional average---the mean payoff objective---is employed. 
In this case, each element contributes equally to the total, and transient elements become less significant as time progresses. 

Shapley, in his seminal work~\cite{Shapley53}, showed that stochastic games played on finite arenas with
discounted objectives are determined in stationary strategies. 
In other words, there are optimal strategies for both players that need only consider the current state of play. 
Bewley and Kohlberg showed that stochastic games with mean payoff objective are also determined in stationary strategies \cite{BewleyKohlberg78} and related the discounted and mean payoff objectives asymptotically \cite{BewleyKohlberg76}.

\paragraph{Discounting the Past.}
Using Shapley's discounting as a conceptual basis, we study
\emph{past-discounting} as a temporally dual notion.
For a finite sequence of weights $\seq{x_0, x_1, \ldots, x_{n}}$, the past-discounted sum with discount factor $\gamma \in [0, 1)$, is defined by the formula $\sum^n_{k=0} \gamma^{n-k} x_k$, and was introduced by Alur, et al.~\cite{alur2012regular}.

To motivate the idea of past-discounting, consider a setting where the weights in a stochastic game arena characterize the performance of an agent.
Under such an interpretation, rewards obtained by the agent at each stage represent the immediate response to their actions from the environment.
In this sense, the total sum of a finite sequence of rewards encountered by the agent provides a quantification of their reputation over that period of time.
This sum may be thought of as a performance evaluation of the agent done by an evaluator considering each past action with equal significance.
In contrast, the past-discounted sum of the finite sequence represents an evaluation of the agent in which actions taken more recently are given greater significance.
Here, the discount factor $\gamma$ can be seen as a parameter encoding the evaluator's preference for recent performance compared to performances occurring long ago.
If we suppose the cause of this preference to be that the evaluator has imperfect memory that worsens over time, then $\gamma$ may represent the rate at which the evaluator forgets or forgives.

If, for instance, the agent is a politician, the reward incurred at each step could correspond to their popularity rating, and the past-discounted sum of these ratings would then be their (weighted) average popularity, placing more significance on recent ratings.
In this scenario, the past-discounted evaluation is more relevant than the total sum evaluation, since favorability of a politician depends on the public's recollection of their past actions, which tends to decay with time.
More generally, past-discounting makes sense whenever the agent in question can learn and adapt and the evaluator either has a time-decaying memory or a time-decaying perception of significance.

We propose the extension of this notion to infinite weight sequences.
Perhaps the most straightforward way to define the past-discounted objective for an infinite sequence $\seq{x_n}_{n\geq0}$ of payoffs is to consider the limit 
$\lim_{n \to \infty} \sum^n_{k=0} \gamma^{n-k} x_k$
of the past-discounted sums of finite prefixes of increasing length.
Even for bounded sequences, however, this summation may fail to converge.
As an example consider the bounded infinite sequence $\seq{1, 2, 1, 2, 1, 2, \ldots}$.
The past-discounted sum for every even step $k = 2n$ is given as 
\[
2 \cdot \frac{1-(\gamma^2)^{n}}{1-\gamma^2} +   \gamma \cdot \frac{1-(\gamma^2)^{n}}{1-\gamma^2},
\]
while, for every odd step $ k = 2n + 1$, we have 
\[
\frac{1-(\gamma^2)^{n}}{1-\gamma^2} + 2\gamma \cdot \frac{1-(\gamma^2)^{n}}{1-\gamma^2}.
\]
Hence, every even subsequence converges to $\frac{2+\gamma}{1-\gamma^2}$, while the odd subsequences converge to $\frac{1+2\gamma}{1-\gamma^2}$, and the limit of the sequence as whole does not exist. 

In light this observation, $\lim_{n \to \infty} \sum^n_{k=0} \gamma^{n-k} x_k$ is ill-posed and does not define a function from bounded sequences to scalars.
Instead, we study the following variants which are well defined over any bounded weight sequence $x = \seq{x_n}_{n \geq 0}$.
\begin{align}
	P_\gamma(x) &\rmdef \liminf_{n \to \infty} \sum^n_{k=0} \gamma^{n-k} x_k \tag{Past-Discounted Payoff} \\
	D_\lambda P_\gamma(x) &\rmdef \lim_{n \to \infty} \sum_{k=0}^n \lambda^k \sum_{i=0}^k \gamma^{k-i} x_i \tag{Discounted Past-Discounted Payoff} \\
	M P_\gamma(x) &\rmdef \liminf_{n \to \infty} \frac{1}{n+1} \sum_{k=0}^n \sum_{i=0}^k \gamma^{k-i} x_i \tag{Mean Past-Discounted Payoff}
\end{align}
The latter two objectives correspond to the Abelian and Ces\`aro summations of the past-discounted weights.
\paragraph{Past-Discounted Payoff (Example).}
Consider the  bounded infinite sequence $x = (3, 4, 5)^\omega$ and the corresponding sequence of past-discounted sums:
\begin{equation*}
	P_\gamma(x) = \seq{3, (3\gamma + 4), (3 \gamma^2+4\gamma+5), (3 \gamma^3+4\gamma^2+5\gamma+3), (3 \gamma^4+4\gamma^3+5\gamma^2+3\gamma+4), \ldots}.
\end{equation*}
It is easy to see that the sequence $P_\gamma(x)$ does not converge.
However, the subsequence containing every third element ($k=3n$) converges to $\frac{3+5\gamma+4\gamma^2}{1-\gamma^3}$.
Likewise, the subsequence with $k=3n+1$ converges to $\frac{4+3\gamma+5\gamma^2}{1-\gamma^3}$ and the subsequence with $k = 3n+2$ converges to $\frac{5+4\gamma+3\gamma^2}{1-\gamma^3}$.

\paragraph{Discounted Past-Discounted Payoff (Example).}
Consider the $\lambda$-discounted sum of sequence $P_\gamma(x)$:
\begin{align*}
	D_\lambda P_\gamma(x) &= \cla{3} + \clb{\lambda (3\gamma  {+} 4}) + \clc{\lambda^2 (3 \gamma^2 {+}4\gamma {+}5}) + \cld{\lambda^3(3 \gamma^3 {+}4\gamma^2 {+}5\gamma {+}3)} + \cle{\lambda^4(3 \gamma^4 {+}4\gamma^3 {+}5\gamma^2 {+}3\gamma+4)} +  \ldots \\
	&= (\cla{3} {+} \clb{3 \gamma\lambda} {+} \clc{3\gamma^2\lambda^2} {+} \ldots) + (\clb{4 \lambda } {+} \clc{4 \lambda^2\gamma} {+}\cld{4\lambda^3\gamma^2} {+}\ldots) + (\clc{5\lambda^2} {+} \cld{5\lambda^3\gamma} {+} \cle{\lambda^5\gamma^2} {+} \ldots) \\
	&\qquad + (\cld{3\lambda^3} {+} \cle{3 \gamma\lambda^4} {+} 3\gamma^2\lambda^5 {+} \ldots) + \ldots \\
	&= 3(1 {+} \gamma\lambda {+} \gamma^2\lambda^2 {+} \ldots) +  4\lambda(1 {+} \lambda\gamma {+} \lambda^2\gamma^2 {+} \ldots) + 5\lambda^2(1  {+} \lambda\gamma {+} \lambda^2\gamma^2 {+} \ldots) \\
	&\qquad + 3\lambda^3(1  {+} \gamma\lambda {+} \gamma^2\lambda^2 {+} \ldots) + \ldots \\
	&= \frac{3+4\lambda+5\lambda^2+3\lambda^3+\ldots}{(1-\gamma\lambda)} \\
	&= \frac{3+4\lambda+5\lambda^2}{(1-\gamma\lambda)(1-\lambda^3)}.
\end{align*}
Notice that the discounted past-discounted sum is equal to to the discounted sum
but with weights scaled by a factor $\frac{1}{1-\gamma\lambda}$.
We show that this is not a mere coincidence, and prove that this equality holds also for Markov chains and
stochastic games.
In doing so, we reduce the problem of solving games with discounted
past-discounted payoff to standard discounted games.

\paragraph{Mean Past-Discounted Payoff (Example).}
Finally, consider the mean payoff of the past-discounted
sequence defined as the $\liminf$ of the following sequence: 
\begin{eqnarray*}
\seq{3, \frac{3\gamma + 4}{2}, \frac{3 \gamma^2+4\gamma+5}{3},  \frac{3 \gamma^3+4\gamma^2+5\gamma+3}{4},
\frac{3 \gamma^4+4\gamma^3+5\gamma^2+3\gamma+4}{5}, \ldots}.
\end{eqnarray*}
Based on the classical Tauberian results, we may be tempted to conjecture that
that mean past-discounted payoff $M P_\gamma(x)$ equals $\lim_{\lambda\upto 1}
(1-\lambda) D_\lambda P_\gamma(x)$, i.e.
\begin{eqnarray*}
 M P_\gamma(x) = 
 \lim_{\lambda\upto 1}
 (1-\lambda) \frac{3+4\lambda+5\lambda^2}{(1-\gamma\lambda)(1-\lambda^3)}
 =  \lim_{\lambda\upto
 1}  \frac{3+4\lambda+5\lambda^2}{(1-\gamma\lambda)(1+\lambda+\lambda^2)}
 =  \frac{3+4+5}{3(1-\gamma)}
\end{eqnarray*}
Indeed, we show that this conjecture holds. Moreover, we show that to compute
the value of mean past-discounted payoff in a stochastic game, it suffices to
scale the weights by a factor of $\frac{1}{1-\gamma}$ to compute the optimal
value for the mean past-discounted games. 

\noindent
\textbf{Contributions.} In summary, the contributions of this paper are itemized below.
\begin{description}
	\item[Past-Discounted Payoff.]
			We study the properties of past-discounted payoffs and show that while they are prefix-independent, they are not submixing (in the sense of \cite{Gimbert07}). Hence, their positionality does not follow from well known results. 
			We prove that past-discounted games are indeed not positionally determined and establish that optimal strategies for past-discounted games may require unbounded memory.
			We consider an approximation method, formulated in terms of sliding windows of finite length over the infinite sequence of rewards, following the work of \cite{ChatterjeeDoyenRandourRaskin13,ChatterjeeDoyenRandourRaskin15,BordaisGuhaRaskin19}, called window past-discounted payoffs and show how to compute their values.
	\item[Discounted Past-Discounted Payoff.]
	        We provide a reduction from the values for concurrent stochastic discounted past-discounted games to those for standard concurrent stochastic discounted games.
			As a result, the computational and strategic complexity are equivalent for these games.
	\item[Mean Past-Discounted Payoff.]
			We provide a reduction from the values of concurrent stochastic mean past-discounted games to those of concurrent stochastic mean payoff games.
			Computational and strategic complexity are therefore equivalent for these games.
		    We prove a Tauberian theorem which asymptotically relates the values of discounted past-discounted games and mean past-discounted games on concurrent stochastic arenas.
\end{description}

\section{Preliminaries}
\label{sec:games}
Let $X^*$ be the set of all finite words and $X^\omega$ be the set of all countably infinite words over a finite set $X$.
A discrete probability distribution over a (possibly infinite) set $X$ is a function  $d : X \to [0, 1]$ such that $\sum_{x \in X} d(x) = 1$ and the support $\supp(d) \rmdef \set{x \in X : d(x) > 0}$ of $d$ is countable.
Let $\dist(X)$ denote the set of all discrete distributions over $X$.
A distribution $d$ is a \emph{point distribution} if $d(x) = 1$ for some $x \in X$.

We consider two-player zero-sum games on finite
stochastic game arenas (SGAs) between two players---player Min and player Max---who
concurrently choose their actions to move a token along the edges of a graph. 
The next state is determined by a probabilistic transition function based on the current state and the players' selected actions.

A \emph{stochastic game arena} (SGA) $G$ is a tuple $(S, A_\Min, A_\Max, w, p)$ where: 
\begin{itemize}
    \item $S$ is a finite set of states,
    \item $A_\Min$ and $A_\Max$ are finite sets of actions for players Min and Max;
    \item $w : S \times A_\Min \times A_\Max \to \reals$ is a weight function, and 
    \item $p : S \times A_\Min \times A_\Max \to \dist(S)$ is a probabilistic transition function. 
\end{itemize}
We write $A_\Min(s) \subseteq A_\Min$ and $A_\Max(s) \subseteq A_\Max$ for the set of actions available to players Min and Max, respectively, at the state $s \in S$.
For states $s, s' \in S$ and actions $(a, b) \in A_\Min(s) \times A_\Max(s)$, we write $p(s'|s, a, b)$ for $p(s, a, b)(s')$.
An SGA is said to be \emph{deterministic} when $p(\cdot \mid s, a, b)$ is a point distribution, for all possible states and action pairs.
An SGA is a \emph{perfect-information} or \emph{turn-based} arena if, for all $s \in S$, at least one of the sets $A_\Min(s)$ and $A_\Max(s)$ is a singleton.
If an SGA is not turn-based it is called \emph{concurrent}.
An SGA where one of the players has only a single choice of action from every state is called a \emph{one-player arena}, otherwise it is considered to be \emph{two-player}.
One-player stochastic arenas are equivalent to Markov decision processes, and one-player deterministic arenas are equivalent to weighted directed graphs.

A \emph{play} on $G$ is an infinite sequence $\pi = \seq{(s_n, a_n, b_n)}_{n \geq 0} \in (S \times A_\Min \times A_\Max)^\omega$ such that $p(s_{k+1} \mid s_k, a_k, b_k) > 0$ for all $k \geq 0$.
A finite play is a sequence in $S \times ((A_\Min \times A_\Max) \times S)^*$.
Let $\subseq{\pi}{k}{n}$ represent the length $n-k$ contiguous subsequence of $\pi$, starting at $(s_k, a_k, b_k)$ and ending at $(s_n, a_n, b_n)$.
Denote by $\last(\pi)$ the final tuple comprising the finite play $\pi$.
We write $\Plays^G$ and $\FPlays^G$, respectively, for the set of all plays and finite plays on the SGA $G$ and $\Plays{}^G(s) $ and $\FPlays^G(s)$ for the respective subsets of these for which $s$ in the initial state.

Starting from an initial state $s\in S$ of an SGA $G$, players $\Min$ and $\Max$ produce an infinite run by concurrently choosing state-dependent actions, and then moving to a successor state determined by the transition function.
A strategy of player Min in $G$ is a function $\sn \colon \FPlays^G \times S \to \dist(A_\Min)$ such that, for all plays $\pi \in \FPlays^G$, we have that $\supp(\sn(\pi, s)) \subseteq A_\Min(s)$.  
A strategy $\sx$ of player Max is defined analogously.
A strategy is \emph{pure} if its image is is a point distribution wherever it is defined; otherwise, it is \emph{mixed}.  
We say that a strategy $\sigma$ is \emph{stationary} if $\sigma(\pi, s) = \sigma(\pi', s)$, for any plays $\pi$ and $\pi'$.
Strategies that are not stationary are called \emph{history dependent}.
A strategy is \emph{positional} if it is pure and stationary.
Let $\Sigma_\Min$ and $\Sigma_\Max$ be the sets of all strategies for the players.
If the arena is not clear from context, we write $\Sigma_\Min^G$ and $\Sigma_\Max^G$.

For an SGA $G$, a state $s$ of $G$, and strategy pair $(\sn,\sx) \in \Sigma_\Min {\times} \Sigma_\Max$, let $\Plays^{\sn, \sx}(s)$ (resp. $\FPlays^{\sn,\sx}(s)$) denote the set of infinite (resp. finite) plays in which player Min and Max
play according to $\sn$ and $\sx$, respectively.  
Given a finite play $\pi \in \FPlays^{\sn, \sx}(s)$, a basic cylinder set $\cyl(\pi)$ is the set of infinite plays in $\Plays^{\sn, \sx}(s)$ for which $\pi$ is a prefix.
Using standard results from probability theory one can construct a probability space $(\Plays^{\sn, \sx}(s), \mathcal{F}^{\sn, \sx}(s), \Pr^{\sn,\sx}_s)$ where $\mathcal{F}^{\sn, \sx}(s)$ is the smallest $\sigma$-algebra generated by the basic cylinder sets and $\Pr^{\sn, \sx}_s : \mathcal{F}^{\sn, \sx}(s) \to [0,1]$ is a unique probability measure such that a finite play $\pi =\seq{(s_k, a_k, b_k)}^n_{k=0}$ has probability $\Pr^{\sn, \sx}_s(\pi) = \prod_{k=0}^n p(s_{k+1} \mid s_k, a_k, b_k) \cdot \nu\left(\subseq{\pi}{0}{k-1}, s_k\right)(a_k) \cdot \chi\left(\subseq{\pi}{0}{k-1}, s_k\right)(b_k)$.

A real-valued random variable $f : \Plays \to \reals$, otherwise known as an \emph{objective function} or synonymously a \emph{payoff function} determines the optimization objective of a particular game.
The expression $\eE^{\sn, \sx}_{s}(f)$ denotes the expectation of $f$ with respect to $\Pr^{\sn, \sx}_s$.
Given a payoff function $f : Plays \to \reals$ and an SGA $G$, the pair $(G, f)$ fully specifies a game in which the goal of Max is to maximize the expected value of $f$ over $G$, while the goal of Min is the opposite.
For every state $s \in S$, define the \emph{upper value} $\UVal{f,s}$ as the minimum payoff player Min can ensure irrespective of player Max's strategy.
Symmetrically, the \emph{lower value} $\LVal{f, s}$ of a state $s\in S$ is the maximum payoff player Max can ensure irrespective of player Min's strategy.
Symbolically,
\begin{eqnarray*}
  \UVal{f, s} \rmdef \inf_{\sn \in \Sigma_\Min} \sup_{\sx\in \Sigma_\Max} \eE^{\sn, \sx}_{s}(f) \quad \text{ and } \quad \LVal{f, s} \rmdef  \sup_{\sx \in \Sigma_\Max} \inf_{\sn \in \Sigma_\Min} \eE^{\sn, \sx}_{s}(f).
\end{eqnarray*}
The inequality $\LVal{f, s} \leq \UVal{f,s}$ holds for all two-player zero-sum stochastic games.  
A game is \emph{determined} when, for every state $s \in S$, the lower value and upper value are equal.
In this case, we say that the value of the game exists with $\Val{f,s} = \LVal{f,s} = \UVal{f,s}$ for every $s\in S$.
If it is not clear from context, we write $\VAL_G$ to specify the arena over which the value is taken.
For strategies $\sn \in \Sigma_\Min$ and $\sx \in \Sigma_\Max$ of players Min and Max, we define their values as
\begin{equation*}
  \VAL^\sn(f, s) \rmdef \sup_{\sx \in \Sigma_\Max} \eE^{\sn, \sx}_{s}(f) \qquad \tn{ and } \qquad \VAL^\sx(f, s) \rmdef \inf_{\sn \in \Sigma_\Min} \eE^{\sn, \sx}_{s} (f).
\end{equation*}
A strategy $\sn_\sharp$ of player Min is called \emph{optimal} if $\VAL^{\sn_\sharp}(f,s) = \Val{f,s}$.
Likewise, a strategy $\sx_\sharp$ of player Max is optimal if $\VAL^{\sx_\sharp}(f,s) = \Val{f,s}$.
We say that a game is \emph{positionally determined} if both players have positional optimal strategies.
Likewise, a game is determined in stationary strategies when both players have optimal stationary strategies.

\begin{definition}
  For an SGA $G$, the following payoffs of player Min to player Max have been considered extensively in literature (See, for example~\cite{FilarVrieze96}).
  \begin{description}
   \item[Limit Payoff.] The limit payoff $L: \Plays \to \reals$ of a play is defined as follows:
    \begin{equation*}
      L \rmdef \seq{(s_n, a_n, b_n)}_{n \geq 0} \mapsto \liminf_{n\to\infty} w(s_n, a_n, b_n).
    \end{equation*}
    \item[Discounted Payoff.] The $\lambda$-discounted payoff $D_\lambda: \Plays \to \reals$ of a play, for a given discount factor $\lambda \in [0, 1)$, is defined as follows:
    \begin{equation*}
      D_\lambda \rmdef \seq{(s_n, a_n, b_n)}_{n \geq 0} \mapsto \lim_{n\to\infty} \sum_{k=0}^{n} \lambda^k w(s_k, a_k, b_k).
    \end{equation*}
    \item[Mean Payoff.] The mean payoff $M: Plays \to \reals$ of a play is given by the long-run average-sum of the weight sequence of the play:
    \begin{equation*}
      M \rmdef \seq{(s_n, a_n, b_n)}_{n \geq 0} \mapsto \liminf_{n\to\infty} \frac{1}{n+1} \sum_{k=0}^{n} w(s_k, a_k, b_k).
    \end{equation*}
  \end{description}
\end{definition}

We refer to SGAs with the limit payoff objective as \emph{stochastic liminf games}, to SGAs with a discounted payoff objective as \emph{stochastic discounted games}, and to SGAs with the mean payoff objective as \emph{stochastic mean payoff games}.

To each game $(G,f)$, there is an associated decision problem called the \emph{threshold problem} or, alternatively, the \emph{value problem}, which is defined relative to a rational number $t \in \bb{Q}$ and an initial state $s$ in the arena.
The threshold problem asks whether the value of the game from state $s$ is bounded below by $t$; that is, to decide if $\Val{f,s} \geq t$.
When referring to the computational complexity of various classes of games or of solving various classes of games, we mean the complexity of deciding the value problem.


The next theorem recalls state of the art results on strategic and computational complexity of stochastic games with liminf and limsup payoff.

\begin{theorem}[Limit Games]
\label{thm:ChatterjeeDoyenHenzinger09}
    All liminf games and limsup games are determined \cite{MaitraSudderth92,MaitraSudderth12}.
    Liminf games and limsup games over one-player deterministic arenas, one-player stochastic arenas, and two-player deterministic arenas are in \Ptime{} \cite{ChatterjeeHenzinger07,ChatterjeeDoyenHenzinger09}.
    Liminf games and limsup games over turn-based two-player stochastic arenas are in $\UP{} \cap \coUP{}$ \cite{GawlitzaSeidl09}.
    Over any type of non-concurrent arena, liminf games and limsup games are positionally determined \cite{ChatterjeeHenzinger07,ChatterjeeDoyenHenzinger09}.
\end{theorem}

The following theorem summarizes state of the art results on strategic and computational complexity of stochastic games with discounted and mean payoff.

\begin{theorem}[Discounted and Mean Payoff Games]
  \label{thm:Shapley53}
  Stochastic discounted games \cite{Shapley53} and stochastic mean payoff games \cite{BewleyKohlberg78} are determined in mixed stationary strategies.
  Both are positionally determined for turn-based game arenas \cite{LiggettLippman69,FilarVrieze96}.
  Stochastic discounted games are in \textnormal{FIXP} and are \textnormal{SQRT-SUM}-hard \cite{EtessamiYannakakis10}.
  Stochastic mean payoff games are in \Exp{} and are \textnormal{PTIME}-hard \cite{ChatterjeeMajumdarHenzinger08}.
  For turn-based SGAs, both types of games are in $\UP \cap \coUP$ \cite{ChatterjeeFijalkow11}.
\end{theorem}
%

Stochastic mean payoff games are intimately connected to stochastic discounted games, as evidenced by the subsequent Tauberian theorem.

\begin{theorem}[Blackwell Optimality\cite{BewleyKohlberg76}]
  \label{thm:BewleyKohlberg}
  As $\lambda$ tends to $1$ from below, the following equation holds for every state $s$:
  \begin{equation*}
    \Val{M, s} = \lim_{\lambda \upto 1} (1 - \lambda) \Val{D_\lambda, s}.
  \end{equation*}
  Moreover, when $\lambda$ is close to 1, strategies optimal for $D_\lambda$ are optimal for $M$. 
\end{theorem}

\section{Past-Discounted Games}
\label{sec:infpast}
\begin{definition}[Past-Discounted Payoffs]
	Let $G = (S, A_\Min, A_\Max, w, p)$ be an SGA, $\pi = \seq{(s_n, a_n, b_n)}_{n \geq 0}$ be an infinite play with weights $\seq{w_n}_{n \geq 0} = \seq{w(s_n, a_n, b_n)}_{n \geq 0}$, and $\gamma \in [0,1)$ be a discount factor.
	For every $n \in \nats$ the $n$-step past-discounted objective is given by the following equation.
	\begin{equation*}
		P^n_\gamma(\pi) \rmdef \sum^n_{k=0} \gamma^{n-k} w_k \tag{Finite Past-Discounted Payoff}
	\end{equation*}
	The lower and upper past-discounted objectives are characterized, respectively, by the following equations.
	\begin{align*}
		\underline{P}_\gamma(\pi) &\rmdef \liminf_{n \to \infty} P^n_\gamma(\pi) = \liminf_{n \to \infty} \sum^n_{k=0} \gamma^{n-k} w_k \tag{Lower Past-Discounted Payoff} \\
		\overline{P}_\gamma(\pi) &\rmdef \limsup_{n \to \infty} P^n_\gamma(\pi) = \limsup_{n \to \infty} \sum^n_{k=0} \gamma^{n-k} w_k \tag{Upper Past-Discounted Payoff}
	\end{align*}
\end{definition}

Since the lower and upper past-discounted payoffs can mutually simulate each other (by negating the sequences of finite-step past-discounted payoffs), we present our results in terms of the lower past-discounted payoff and simply refer to it as the past-discounted payoff written as $P_\gamma$.
This is without loss of generality and our results apply to the upper past-discounted payoff as well.

\subsection{Properties of the Past-Discounted Payoff}

In \cite{Gimbert07} Gimbert provides a convenient sufficient condition on objectives that ensures positional determinacy in some situations.
If an objective function is both prefix-independent and submixing, as defined below, then there are optimal positional strategies for the objective over stochastic one-player arenas.

\begin{definition}
	\label{def:prefix-independence-submixing}
	Suppose that $f$ is an arbitrary objective function and that $x = \seq{x_n}_{n \in \nats}$ and $y = \seq{y_n}_{n \in \nats}$ are infinite sequences of finite words over the finite set of possible weights $W$, ie. $x_n, y_n \in W^*$ for all $n$.
	Define the shuffle of these sequences as $x \shuffle y = \seq{x_1,y_1,x_2,y_2,\ldots}$.
	The objective $f$ is
	\begin{itemize}
		\item \emph{prefix-independent} if $f(u v) = f(v)$,	for any sequences $u \in W^*$ and $v \in W^\omega$,
		\item \emph{submixing} if $f(x \shuffle y) \leq \max \set{f(x), f(y)}$.
	\end{itemize}
\end{definition}

\begin{proposition}
	\label{prop:prefix-independent}
	The past-discounted payoffs are prefix-independent.
\end{proposition}
\begin{proof}
	Suppose that $x = \seq{x_0, \ldots, x_m}$ a finite sequence of weights, $y = \seq{y_n}_{n \geq 0}$ is an infinite sequence of weights, and let $xy = \seq{x_0, \ldots, x_m, y_1, y_2, \ldots}$ be the infinite sequence created by prepending $x$ to $y$.
	For any $n$ greater than $m$, the $n$-step finite past-discounted payoff on $xy$ may be written as
	\begin{equation*}
		P^n_\gamma(xy) = \sum^n_{k=m} \gamma^{n-k} y_{k-m} + \sum^m_{k=0} \gamma^{n-k} x_k ,
	\end{equation*}
	by separating the portions of the summation concerning elements from $x$ and $y$.
	Manipulating the summation indices, we may rewrite this equation as
	\begin{equation*}
		P^n_\gamma(xy) = \sum^{n-m}_{k=0} \gamma^{n-m-k} y_k + \sum^m_{k=0} \gamma^{n-k} x_k.
	\end{equation*}
	Factoring $\gamma^{n-m}$ from the summation of $x$ elements yields the equality $\sum^m_{k=0} \gamma^{n-k} x_k = \gamma^{n-m} \sum^m_{k=0} \gamma^{m-k} x_k$, and this implies the equation
	\begin{equation*}
		P^n_\gamma(xy) = P^{n-m}_\gamma(y) + \gamma^{n-m} P^m_\gamma(x).
	\end{equation*}
	Since $m$ is fixed and finite, $\liminf_{n \to \infty} \gamma^{n-m}P^m_\gamma(x) = 0$.
	Taking the limit inferior of both sides of this equation yields that $P_\gamma(xy) = P_\gamma(y)$, and thus we conclude that $P_\gamma$ is prefix-independent. 
\end{proof}

\begin{proposition}
	\label{prop:not_submixing}
	The past-discounted payoffs are not submixing.
\end{proposition}
\begin{proof}
	\newcommand{\gam}{0.1}
	Consider the sequence $z = \seq{200, 2, 100, 1}^\omega$ as the shuffle of the sequences $x = \seq{2, 1, 200, 100}^\omega$ and $y = \seq{200, 100, 2, 1}^\omega$, which is valid since $z$ can be produced as an interleaving of the elements of $x$ and $y$ that preserves the relative ordering of the terms from each, as illustrated below.
	\begin{equation*}
		\resizebox{0.75\textwidth}{!}
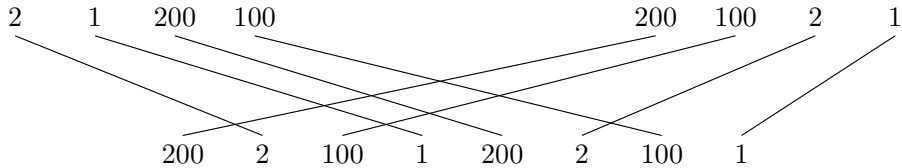
{
			\begin{tikzpicture}[node distance = 1cm]
				\node (0) {200};
				\node (1) [right of = 0] {2};
				\node (2) [right of = 1] {100};
				\node (3) [right of = 2] {1};
				\node (4) [right of = 3] {200};
				\node (5) [right of = 4] {2};
				\node (6) [right of = 5] {100};
				\node (7) [right of = 6] {1};

				\node (a) [above left = 1.25cm and 1.5cm of 0] {2};
				\node (b) [right of = a] {1};
				\node (c) [right of = b] {200};
				\node (d) [right of = c] {100};

				\node (h) [above right = 1.25cm and 1.5cm of 7] {1};
				\node (g) [left of = h] {2};
				\node (f) [left of = g] {100};
				\node (e) [left of = f] {200};

				\path (0.north) edge (e.south);
				\path (1.north) edge (a.south);
				\path (2.north) edge (f.south);
				\path (3.north) edge (b.south);
				\path (4.north) edge (c.south);
				\path (5.north) edge (g.south);
				\path (6.north) edge (d.south);
				\path (7.north) edge (h.south);
			\end{tikzpicture}
		}
	\end{equation*}
	Because these sequences are periodic, we get that $P_\gamma(z)$ is equal to
	\begin{equation*}
		\min\set{\frac{200\gamma^3 + 2\gamma^2 + 100 \gamma + 1}{1 - \gamma^4}, \frac{\gamma^3 + 200\gamma^2 + 2\gamma + 100}{1 - \gamma^4}, \frac{100\gamma^3 + \gamma^2 + 200\gamma + 2}{1 - \gamma^4}, \frac{2\gamma^3 + 100\gamma^2 + \gamma + 200}{1 - \gamma^4}},
	\end{equation*}
	and $P_\gamma(x) = P_\gamma(y)$ are equal to
	\begin{equation*}
		\min\set{\frac{2\gamma^3 + \gamma^2 + 200\gamma + 100}{1 - \gamma^4}, \frac{100\gamma^3 + 2\gamma^2 + \gamma + 200}{1 - \gamma^4}, \frac{200\gamma^3 + 100\gamma^2 + 2\gamma + 1}{1 - \gamma^4}, \frac{\gamma^3 + 200\gamma^2 + 100\gamma + 2}{1 - \gamma^4}}.
	\end{equation*}
	Setting $\gamma = \frac{1}{10}$, these expressions evaluate approximately to $P_\gamma(z) \approx 11.2211$ and $P_\gamma(x) = P_\gamma(y) \approx 2.4002$.
	Therefore $P_\gamma(z) > \max\set{P_\gamma(x), P_\gamma(y)}$, establishing that $P_\gamma(z)$ is not submixing for every discount factor.	
\end{proof}

\begin{wrapfigure}{r}{0.4\textwidth}
	\vspace{-1em}
	\centering
	\scalebox{0.65}{
	\begin{tikzpicture}
		\begin{semilogyaxis}[xlabel=$\gamma$, xmin=0, xmax=1, ymin=1, ymax=10000, legend pos=north west]\
			\addplot[domain=0:1, samples=100, color=red]{
				min((100 + 200*x + x^2 + 2*x^3)/(1 - x^4),
				(200 + x + 2*x^2 + 100*x^3)/(1 - x^4),
				(1 + 2*x + 100*x^2 + 200*x^3)/(1 - x^4),
				(2 + 100*x + 200*x^2 + x^3)/(1 - x^4))
			};
			\addlegendentry{$\max\set{P_\gamma(x),P_\gamma(y)}$}
			\addplot[domain=0:1, samples=100, color=blue]{
				min((200 + x + 100*x^2 + 2*x^3)/(1 - x^4),
				(2 + 200*x + x^2 + 100*x^3)/(1 - x^4),
				(100 + 2*x + 200*x^2 + x^3)/(1 - x^4),
				(1 + 100*x + 2*x^2 + 200*x^3)/(1 - x^4))
			};
			\addlegendentry{$P_\gamma(z)$}
		\end{semilogyaxis}
	\end{tikzpicture}
	}
	\vspace{-2em}
	\label{fig:not_submixing}
\end{wrapfigure}
In fact, the sequences $x,y,z$ from the proof of \Cref{prop:not_submixing} provide a counterexample to submixing of $P_\gamma$, for all discount factors in the unit interval, other than those very close to 1.
This is illustrated by the figure to the right, which displays, using the sequences from the proof of \Cref{prop:not_submixing}, $\max\set{P_\gamma(x), P_\gamma(y)}$ in red and $P_\gamma(z)$ in blue, for each possible discount factor $\gamma \in [0,1)$.

Our next result settles the matter of positional determinacy for past-discounted games by establishing that optimal strategies may require unbounded memory in even the simplest of contexts: games on one-player deterministic arenas.

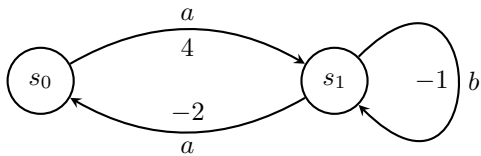
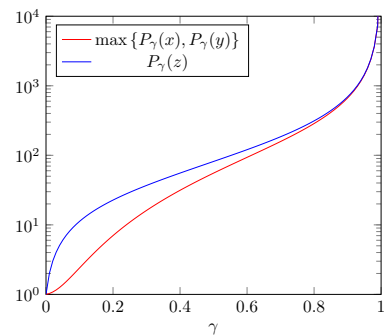
\begin{figure}[t]
	\vspace{-3em}
	\centering
	\begin{tikzpicture}[> = stealth, thick]
		\node[state] (0) {$s_0$};
		\node[state] (1) [right = 3cm of 0] {$s_1$};

		\path[->] (0) edge[bend left] node[above] {$a$} node[below] {$4$} (1);
		\path[->] (1) edge[bend left] node[below] {$a$} node[above] {$-2$} (0);
		\path[->] (1) edge[loop right] node {$b$} node[left] {$-1$} (1);
	\end{tikzpicture}
	\vspace{-3em}
	\caption{SGA where optimal strategies for $P_\gamma$ require unbounded memory.}
	\label{fig:non-positional}
\end{figure}

\begin{theorem}
	\label{thm:infinite_memory}
	Past-discounted games are not determined in stationary strategies and optimal strategies may require unbounded memory.
\end{theorem}
\begin{proof}
	\newcommand{\gam}{0.5}
	We proceed by exhibiting a one-player deterministic arena on which the optimal strategy, at any point in time, depends upon the entire history of play.
	Without loss of generality, we fix the single player in this context as \Min{} whose goal is to minimize the past-discounted payoff.
	Consider the one-player deterministic arena shown in \Cref{fig:non-positional} on which there are exactly two positional strategies.
	One strategy chooses action $a$ from state $s_1$ while the other selects action $b$.
	These strategies generate plays $\pi_a$ and $\pi_b$, respectively, with values
	\begin{equation*}
		P_\gamma(\pi_a) = \min\set{\frac{4 - 2\gamma}{1 - \gamma^2}, \frac{-2 + 4\gamma}{1 - \gamma^2}} \qquad \tn{ and } \qquad P_\gamma(\pi_b) = \frac{-1}{1 - \gamma}
	\end{equation*}
	If $\gamma = \frac{1}{2}$, then $P_\gamma(\pi_a) = \fpeval{min((4 - 2*\gam) / (1 - \gam*\gam), (-2 + 4*\gam) / (1 - \gam*\gam))}$ and $P_\gamma(\pi_b) = -2$.
	A play $\pi_\sharp$ can achieve a better value, namely $P_\gamma(\pi_\sharp) = -2 - \frac{\gamma}{1 - \gamma} = -3$,
	if it is consistent with a strategy that chooses action $a$ from state $s_1$ infinitely often and also eventually selects $n$ consecutive $b$ actions, for every $n \geq 0$.
	Any strategy satisfying these conditions generates a play $\pi_\sharp$ with weight sequence $\seq{w_n}_{n \geq 0}$ where, for every $n \in \nats$, there exists $m \in \nats$ with $n \leq m$, such that $w_m = -2$ and $w_{m-k} = -1$, for $1 \leq k \leq n$.
	As a result, it holds that 
	\begin{equation*}
		P^m_\gamma(\pi_\sharp) = -2 - \frac{\gamma(1 - \gamma^{n})}{1 - \gamma} + \gamma^{n+1} P^{m-n}_\gamma(\pi_\sharp),
	\end{equation*}
	and we get that $P_\gamma(\pi_\sharp) = -2 - \frac{\gamma}{1 - \gamma} = -3$ in the limit when $\gamma = \frac{1}{2}$.
	
	Since a strategy generating $\pi_\sharp$ must choose distinct actions from the same state at various points in the game, it cannot be positional.
	Because mixed stationary strategies are weighted combinations of positional strategies, their expected values fall between those of the positional strategies, and thus do not achieve better values.
	Therefore past-discounted games are not determined in stationary strategies, even on one-player deterministic arenas.
	Since all other types of arenas under consideration generalize this type, it follows that past-discounted games are not determined in stationary strategies over any type of arena.
	Moreover, an optimal strategy must be able to generate all optimal $n$-step subsequences have been achieved thus far in order to determine the number of times action $b$ should be pumped the next time the player is in state $s_1$.
	Therefore, we conclude that there is no optimal strategy with finitely bounded memory. 
\end{proof}

\subsection{Window Past-Discounted Games}
\label{sec:window-past}
We now introduce window past-discounted games as a means to approximate the values of past-discounted games.
This payoff is based on the sequence of finite past-discounted sums of fixed length and considers only the last $\ell$ elements of the play at any given point.
The value of $\ell$ is considered as the \emph{window length}.

\begin{definition}
	Let $G = (S, A_\Min, A_\Max, w, p)$ be an SGA, $\pi = \seq{(s_n, a_n, b_n)}_{n \geq 0}$ be an infinite play with weight sequence $\seq{w_n}_{n \geq 0} = \seq{w(s_n, a_n, b_n)}_{n \geq 0}$, and $\ell \in \nats$ be a window length.
	The window-past discounted payoff is characterized by the following equation.
	\begin{equation*}
		W_\ell P_\gamma(\pi) \rmdef \liminf_{n \to \infty} P_\gamma^\ell(\subseq{\pi}{n-\ell}{\infty}) = \liminf_{n \to \infty} \sum^n_{k=\max(0,n-\ell)} \gamma^{n-k} w_k \tag{Window Past-Discounted Payoff}
	\end{equation*}
\end{definition}

The following proposition formalizes the relationship between  past-discounted payoffs, window past-discounted payoffs, and the liminf payoff.

\begin{proposition}
    We have the following equalities.
    \begin{enumerate}
    	\item $\lim_{\ell \to \infty} W_\ell P_\gamma = P_\gamma$.
    	\item If $\ell = 0$, then $W_\ell P_\gamma = L$.
    \end{enumerate}
\end{proposition}

Our next result establishes determinacy of window past-discounted games by characterizing their values in terms of the values of certain liminf games.

\begin{theorem}
	\label{thm:window_liminf}
	Window past-discounted games are determined.
	For any window past-discounted game $(G, W_\ell P_\gamma)$, there is a liminf game $(H, L)$ such that, for every state $s$ in $G$, there exists a state $s'$ in $H$ where
	\begin{equation*}
		\VAL_G(W_\ell P_\gamma, s) = \VAL_H(L, s').
	\end{equation*}
\end{theorem}
\begin{proof}
	We construct an arena $H$ for a liminf game $(H, L)$ as the product $H = G \times \M^\ell$, where $\M^\ell = (\Sigma, Q, q_0, \delta)$ is a deterministic finite automaton such that
	\begin{itemize}
		\item $\Sigma = S \times A_\Min \times A_\Max$,
		\item $Q = \set{q_0} \cup \set{\pi \in \FPlays^G : |\pi| \leq \ell}$,
		\item $\delta(q_0, sab) = sab$,
		\item $\delta(s_0 a_0 b_0 \ldots s_{k-1} a_{k-1} b_{k-1}, sab) = \begin{cases}
			s_0 a_0 b_0 \ldots s_{k-1} a_{k-1} b_{k-1} sab &\tn{if } k < \ell, \\
			s_1 a_1 b_1 \ldots s_{k-1} a_{k-1} b_{k-1} sab  &\tn{if } k = \ell.
		\end{cases}$
	\end{itemize}
	More precisely, $H = G \times \M^\ell = (S^\ell, A^\ell_\Min, A^\ell_\Max, w^\ell, p^\ell)$ where
	\begin{itemize}
		\item $S^\ell = \set{(s, q) \in S \times Q : q = q_0 \tn{ or } \last(q) = s_ka_kb_k \tn{ and } p(s \mid s_k, a_k, b_k) > 0}$,
		\item $A^\ell_\Min(s,q) = A_\Min(s)$ and $A^\ell_\Max(s,q) = A_\Max(s)$,
		\item $w^\ell((s, s_0 a_0 b_0 \ldots s_{k-1} a_{k-1} b_{k-1}), a, b) = w(s, a, b) + \sum^{k-1}_{i=0} \gamma^{k-i} w(s_i, a_i, b_i)$,
		\item $p^\ell((s',q') \mid (s,q), a, b) = \begin{cases}
			p(s' \mid s, a, b) &\tn{if } \delta(q, sab) = q', \\
			0 &\tn{otherwise}.
		\end{cases}$
	\end{itemize}
	The number of states in $H$ is on the order of $|S^\ell| = O(|S|2^\ell)$. \newline
	

	\noindent
	\textit{Claim: for every strategy pair $\nu, \chi$ over $G$ there is a strategy pair $\iota, \alpha$ over $H$ such that $\bb{E}^{\nu, \chi}_s(W_\ell P_\gamma) = \bb{E}^{\iota,\alpha}_{(s, q_0)}(L)$, for all states $s$ in $G$.}\newline
	
	Consider $\pi \in \Plays^G$ and the finite play
	\begin{equation*}
		\subseq{\pi}{0}{n} = \seq{(s_0, a_0, b_0),\ldots,(s_\ell, a_\ell, b_\ell),\ldots,(s_{n-\ell}, a_{n-\ell}, b_{n-\ell}),\ldots,(s_n, a_n, b_n)}.
	\end{equation*}
	We know from the construction above that, for every $1 \leq k \leq \ell$, there exists a state $(s_k, \subseq{\pi}{0}{k-1}) = (s_k, s_0 a_0 b_0 \ldots s_{k-1} a_{k-1} b_{k-1})$ in $H$ and for $k=0$ we can associate the state $(s_0, q_0)$ in $H$.
	Likewise, for every $\ell < k \leq n$, there exists a state $(s_k, \subseq{\pi}{k-\ell}{k})$ in $H$.
	From the definition of the transition function $\delta$ of $\M^\ell$ and the state set $S^\ell$ of $H$, it follows that $\subseq{\pi}{0}{n} \in \FPlays^G$ implies the existence of $\subseq{\varpi}{0}{n} \in \FPlays^H$, where
	\begin{align*}
		\subseq{\varpi}{0}{n} =& \Big\langle((s_0, q_0), a_0, b_0), ((s_1, \subseq{\pi}{0}{0}), a_1, b_1), \ldots, ((s_\ell, \subseq{\pi}{0}{\ell-1}), a_\ell, b_\ell), \ldots \\
		&\qquad \ldots, ((s_{n-\ell}, \subseq{\pi}{n-2\ell-1}{n-\ell-1}), a_{n-\ell}, b_{n-\ell}), \ldots, ((s_n, \subseq{\pi}{n-\ell-1}{n-1}), a_n, b_n) \Big\rangle.
	\end{align*}
	If $\subseq{\pi}{0}{n} \in \FPlays^{\nu,\chi}$ for arbitrary strategies $\nu, \chi$, then it follows from the definition of $\Pr^{\nu,\chi}_{s_0}$ that
	\begin{equation*}
		\Pr^{\nu,\chi}_{s_0}(\subseq{\pi}{0}{n}) = \prod^n_{k=1} p(s_k \mid s_{k-1}, a_{k-1}, b_{k-1}) \cdot  \nu\left(\subseq{\pi}{0}{k-1}, s_k\right)(a_k) \cdot \chi\left(\subseq{\pi}{0}{k-1}, s_k\right)(b_k).
	\end{equation*}
	Likewise, if $\subseq{\varpi}{0}{n} \in \FPlays^{\iota, \alpha}$ for arbitrary strategies $\iota, \alpha$, then we have that
	\begin{equation*}
		\Pr^{\iota, \alpha}_{(s_0, q_0)}(\subseq{\varpi}{0}{n}) = \prod^n_{k=1} p^\ell\left((s_k, q_k) \mid (s_{k-1}, q_{k-1}), a_{k-1}, b_{k-1}\right) \cdot \iota\left(\subseq{\varpi}{0}{k-1}, (s_k, q_k)\right)(a_k) \cdot \alpha\left(\subseq{\varpi}{0}{k-1}, (s_k, q_k)\right)(b_k).
	\end{equation*}
	By definition of the probabilistic transition function $p^\ell$ of $H$, we have the equality
	\begin{equation*}
		p(s_k \mid s_{k-1}, a_{k-1}, b_{k-1}) = p^\ell\left((s_k, q_k) \mid (s_{k-1}, q_{k-1}), a_{k-1}, b_{k-1}\right).
	\end{equation*}
	Therefore, $\Pr^{\nu,\chi}_{s_0}(\subseq{\pi}{0}{n}) = \Pr^{\iota, \alpha}_{(s_0, q_0)}(\subseq{\varpi}{0}{n})$, for all $n \in \nats$, precisely when
	\begin{equation}
	\label{eq:strategy_bijection}
	\begin{aligned}
		\nu(\subseq{\pi}{0}{n}, s) &= \iota\left(\subseq{\varpi}{0}{n}, \left(s, \subseq{\pi}{\max(0, n-1-\ell)}{n-1}\right)\right), \\
		\chi(\subseq{\pi}{0}{n}, s) &= \alpha\left(\subseq{\varpi}{0}{n}, \left(s, \subseq{\pi}{\max(0, n-1-\ell)}{n-1}\right)\right)
		\end{aligned}
	\end{equation}
	hold for all $n \in \nats$.
	These equations facilitate translation between arbitrary strategy pairs over arenas $G$ and $H$, and so there indeed exist strategies in $H$ for any strategies in $G$, and vice versa, that induce equivalent probability measures.
	In combination with the definition of $w^\ell$, this implies the equality of expectations for the strategies $\nu,\chi$ on $G$ with that for the strategies $\iota, \alpha$ on $H$.
	To see why, expand out the definition of $\bb{E}^{\nu,\chi}_s(W_\ell P_\gamma)$ to get the following.
	\begin{equation*}
		\bb{E}^{\nu,\chi}_s(W_\ell P_\gamma) = \bb{E}^{\nu,\chi}_s \left(\liminf_{n \to \infty} \sum^{n}_{k = \max(0, n-\ell)} \gamma^{n-k} w(s_k, a_k, b_k) \right)
	\end{equation*}
	Expanding the definition of $\bb{E}^{\iota,\alpha}_{(s,q_0)}(L)$ yields an almost identical result.
	\begin{equation*}
		\bb{E}^{\iota, \alpha}_{(s,q_0)}(L) = \bb{E}^{\iota, \alpha}_{(s,q_0)} \left(\liminf_{n \to \infty} w^\ell(s_k, a_k, b_k)\right) = \bb{E}^{\iota, \alpha}_{(s,q_0)} \left(\liminf_{n \to \infty} \sum^{n}_{k = \max(0, n-\ell)} \gamma^{n-k} w(s_k, a_k, b_k) \right)
	\end{equation*}
	Because of the equivalence we have established between the probability measures $\Pr^{\nu,\chi}_s$ and $\Pr^{\iota,\alpha}_{(s, q_0)}$, assuming that $\nu,\iota$ and $\chi,\alpha$ satisfy \cref{eq:strategy_bijection}, it follows that these expectations are equivalent as well:
	\begin{equation*}
		\bb{E}^{\nu,\chi}_s(W_\ell P_\gamma) = \bb{E}^{\iota, \alpha}_{(s,q_0)}(L).
	\end{equation*}

	Since we have proven the above claim and established, in \cref{eq:strategy_bijection}, a bijection between strategy pairs that preserves expectation, we know that\footnote{See \Cref{sec:window_addendum} for further details.}
	\begin{align*}
	   \inf_\nu \sup_\chi \bb{E}^{\nu,\chi}_s(W_\ell P_\gamma) &= \inf_\iota \sup_\alpha \bb{E}^{\iota,\alpha}_{(s,q_0)}(L), \\
	   \sup_\chi \inf_\nu \bb{E}^{\nu,\chi}_s(W_\ell P_\gamma) &= \sup_\alpha \inf_\iota \bb{E}^{\iota,\alpha}_{(s,q_0)}(L).
	\end{align*}
	Furthermore, we know liminf games are determined, allowing us to infer that
	\begin{equation*}
	    \inf_\nu \sup_\chi \bb{E}^{\nu,\chi}_s(W_\ell P_\gamma) = \sup_\chi \inf_\nu \bb{E}^{\nu,\chi}_s(W_\ell P_\gamma) = \VAL_H(L, (s,q_0)),
	\end{equation*}
	from which it follows that window past-discounted games are determined with 
	\begin{equation*}
	    \VAL_G(W_\ell P_\gamma, s) = \VAL_H(L, (s,q_0)).
	\end{equation*}
\end{proof}


From \Cref{thm:window_liminf,thm:ChatterjeeDoyenHenzinger09} we obtain the following corollary on the strategic and computational complexity of window past-discounted games.

\begin{corollary}[Complexity of Window Past-Discounted Games]
	\leavevmode
	\begin{enumerate}
		\item For any window past-discounted game over a turn-based arena, both players have pure optimal strategies requiring $O(\ell)$ memory.
		\item Window past-discounted games are in \Exp{} over stochastic one-player arenas and deterministic turn-based two-player arenas.
		For fixed window length, these games are in \Ptime{}.
		\item Turn-based stochastic window past-discounted games are in $\NExp{} \cap \coNExp{}$, and are in $\UP \cap \coUP$ for fixed window length.
	\end{enumerate}
\end{corollary}

\begin{remark}
	Algorithmic aspects of liminf games and limsup games over concurrent two-player arenas have not been extensively studied as far as we know (cf. Section 4.7 of \cite{ChatterjeeHenzinger08}).
	As a result, strategic complexity and computational complexity of the these games on concurrent arenas is unknown.
	Resolving these questions would also provide complexity bounds for window past-discounted games over concurrent arenas.
\end{remark}

\section{Discounted Past-Discounted Games}
\label{sec:future-past}
\begin{definition}
  Let $G = (S, A_\Min, A_\Max, w, p)$ be an SGA, $\pi = \seq{(s_n, a_n, b_n)}_{n \geq 0}$ be an infinite play with weight sequence $\seq{w_n}_{n \geq 0} = \seq{w(s_n, a_n, b_n)}_{n \geq 0}$, and $\lambda,\gamma \in [0,1)$ be discount factors.
  The discounted past-discounted payoff is characterized by the following equation.
  \begin{equation*}
    D_\lambda P_\gamma(\pi) \rmdef \lim_{n \to \infty} \sum^n_{k=0} \lambda^k P^k_\gamma(\pi) = \lim_{n\to\infty} \sum_{k=0}^{n} \lambda^k \sum_{i=0}^{k} \gamma^{k-i} w_i \tag{Discounted Past-Discounted Payoff}
  \end{equation*}
\end{definition}

We now show that discounted past-discounted games are determined in stationary strategies by reducing the problem to computing optimal strategies for a standard discounted game on the same arena. 
Our primary tool in showing the reduction is a classical theorem from real analysis (see \cite{Rudin64}, for instance). 

\begin{theorem}[Mertens' Theorem]
  Let $a = \seq{a_n}_{n \geq 0}$ and $ b = \seq{b_n}_{n \geq 0}$ be two sequences, and let $c = \seq{c_n}_{n \geq 0}$ be the Cauchy product $c_n = \sum_{k=0}^{n} a_k b_{n-k}$ of $a$ and $b$.
  If $\lim_{n\to\infty} \sum_{k=0}^{n} a_k = A$ and $\lim_{n\to\infty} \sum_{k=0}^{n} b_k = B$ and at least one of the sequences $\seq{\sum_{k=0}^{n} a_k}_{n \geq 0}$ and $\seq{\sum_{k=0}^{n} a_k}_{n \geq 0}$ converges absolutely, then the limit of the Cauchy product exists such that $\lim_{n\to\infty} \sum_{n=0}^{N} c_n = A B$. 
\end{theorem}

\begin{theorem}
  \label{thm:pd-main}
  Given a state $s$ of some SGA and discount factors $\gamma, \lambda \in [0, 1)$, it holds that
  \begin{equation*}
      \Val{D_\lambda P_\gamma, s} = \frac{1}{1 - \gamma\lambda}\Val{D_\lambda, s}.
  \end{equation*}
\end{theorem}
\begin{proof}
  To prove this theorem, it suffices to show that $\eE^{\sn, \sx}_{s} (D_\lambda P_\gamma) = \frac{1}{1-\gamma\lambda}\eE^{\sn, \sx}_{s} (D_\lambda)$ holds for any pair of strategies $(\sn, \sx) \in \Sigma_\Min \times \Sigma_\Max$ and consistent play $\pi = \seq{(s_n, a_n, b_n)}_{n \geq 0}$.
  In turn, this reduces to proving the equation $D_\lambda P_\gamma(\pi) = \frac{1}{1-\gamma\lambda} D_\lambda(\pi)$.
  This can be done by observing the following sequence of equalities, in which $w_k = w(s_k, a_k, b_k)$:
  \begin{align}
    D_\lambda P_\gamma(\pi) &= \lim_{n \to \infty} \sum_{k=0}^{n}  \lambda^k \sum_{i=0}^{k} \gamma^{k-i} w_i \\
    &= \lim_{n \to \infty} \sum_{k=0}^{n} \sum_{i=0}^{k} \lambda^k \gamma^{k-i} w_i \\
    &= \lim_{n \to \infty} \sum_{k=0}^{n} \sum_{i=0}^{k}
    (\gamma\lambda)^{k-i} \lambda^{i} w_i \\
    &= \left( \lim_{n\to\infty} \sum_{k=0}^{n} (\gamma\lambda)^{k} \right) \left( \lim_{n\to\infty} \sum_{k=0}^{n} \lambda^{k} w_k \right) \\
    &= \frac{D_\lambda(\pi)}{1-\gamma\lambda} .
  \end{align}
  The equality (2) is from the definition of $D_\lambda P_\gamma$.
  Equalities (3) and (4) follow from basic algebra.
  The equality (5) follows from Mertens' theorem, since $\seq{\sum_{i=0}^{k} (\gamma\lambda)^{k-i} \lambda^i w_i}_{k\geq 0}$ is the Cauchy product of $\seq{(\gamma\lambda)^k}_{k\geq 0}$ and $\seq{\lambda^k w_k}_{k\geq 0}$.
  Step (6) is a consequence of the definition of geometric series and $D_\lambda$. 
\end{proof}

\begin{corollary}
  For discounted past-discounted stochastic games:
  \begin{enumerate}
    \item the value is determined in mixed stationary strategies over concurrent arenas;
    \item the value is determined in positional strategies over turn-based arenas;
    \item the complexity of the value problem is equivalent to that of discounted games.
  \end{enumerate}
\end{corollary}

\section{Mean Past-Discounted Games}
\label{sec:avg-past}
\begin{definition}
  Let $G = (S, A_\Min, A_\Max, w, p)$ be an SGA, $\pi = \seq{(s_n, a_n, b_n)}_{n \geq 0}$ be an infinite play with weight sequence $\seq{w_n}_{n \geq 0} = \seq{w(s_n, a_n, b_n)}_{n \geq 0}$, and $\gamma \in [0,1)$ be a discount factor.
  The mean past-discounted payoff is characterized by the following equation.
  \begin{equation*}
    M P_\gamma(\pi) \rmdef \liminf_{n \to \infty} \frac{1}{n+1} \sum^n_{k=0} P^k_\gamma(\pi) = \liminf_{n\to\infty} \frac{1}{n+1} \sum_{k=0}^{n} \sum_{i=0}^{k} \gamma^{k-i} w_i \tag{Mean Past-Discounted Payoff} 
  \end{equation*}
\end{definition}

Taking a similar approach to \Cref{sec:future-past}, we show that mean past-discounted games are determined in stationary strategies.
The proof, however, uses alternate techniques, as Mertens' theorem fails to apply.

\begin{theorem}
  \label{thm:pm-main}
  Given a state $s$ of some SGA and discount factor $\gamma \in [0,1)$ it holds that
  \begin{equation*}
      \Val{M P_\gamma, s} = \frac{1}{1 - \gamma}\Val{M, s}.
  \end{equation*}
\end{theorem}
\begin{proof}
  To prove this theorem, it suffices to show that $\eE^{\sn, \sx}_{s} (M P_\gamma) = \frac{1}{1-\gamma}\eE^{\sn, \sx}_{s} (M)$ for every pair of strategies $(\sn, \sx) \in \Sigma_\Min \times \Sigma_\Max$ and consistent play $\pi = \seq{(s_n, a_n, b_n)}_{n \geq 0}$.
  This amounts to establishing the equation $M P_\gamma(\pi) = \frac{1}{1-\gamma} M(\pi)$, which is shown in the following sequence of equalities, where $w_k = w(s_k, a_k, b_k)$:

\begin{align}
    M P_\gamma(\pi) &= \liminf_{n\to\infty} 
    \frac{1}{n+1} \sum_{k=0}^{n} \sum_{i=0}^{k} \gamma^{k-i} w_i \\
    &= \liminf_{n\to\infty} \sum^n_{k=0} \frac{w_k (1 - \gamma^{n+1-k})}{(n+1) (1-\gamma)} \\
    &= \liminf_{n\to\infty} \left( \sum^n_{k=0} \frac{w_k}{(n+1)(1-\gamma)} - \sum^n_{k=0} \frac{\gamma^{n+1-k} w_k}{(n+1) (1-\gamma)} \right) \\
    &= \liminf_{n\to\infty} \sum^n_{k=0} \frac{w_k}{(n+1)(1-\gamma)} \\
    &= \frac{M(\pi)}{1-\gamma}.
  \end{align}
  The first and the last equalities, (7) and (11), follow from definition of the payoff functions $M P_\gamma$ and $M$.
  The equality (8) is justified by the proof of \Cref{lem:helper1}.
  Going from (8) to (9) is done via basic algebra.
  The equality (10) follows from \Cref{lem:helper2} and the super-additivity of $\liminf$: for infinite sequences $\seq{a_n}_{n \geq 0}$ and $\seq{b_n}_{n \geq 0}$ the inequality $\liminf_{n \to \infty} (a_n + b_n) \geq \liminf_{n \to \infty} a_n  + \liminf_{n \to \infty} b_n$ holds, and if either sequence converges, then $\liminf_{n \to \infty} (a_n + b_n) = \liminf_{n \to \infty} a_n + \liminf_{n \to \infty} b_n$. 
\end{proof}

\begin{lemma}
\label{lem:helper1} 
For every play $\pi = \seq{(s_n, a_n, b_n)}_{n \geq 0}$ and prefix of length $n \geq 0$
\begin{equation*}
    \label{eq:induct}
    \frac{1}{n+1} \sum_{k=0}^{n} \sum_{i=0}^{k} \gamma^{k-i} w(s_i, a_i, b_i) = \sum^n_{k=0} \frac{w(s_k, a_k, b_k) (1 - \gamma^{n+1-k})}{(n+1) (1-\gamma)}.
\end{equation*}
\end{lemma}
\begin{proof}
We proceed by induction on $n$. (see \Cref{sec:proof_l1})
\end{proof}

\begin{lemma}
\label{lem:helper2} 
For every play $\pi = \seq{(s_n, a_n, b_n)}_{n \geq 0}$ the following equation holds: 
\begin{equation*}
    \liminf_{n\to\infty} \sum^n_{k=0} \frac{\gamma^{n+1-k} w(s_k, a_k, b_k)}{(n+1)(1-\gamma)} = 0.
\end{equation*}
\end{lemma}   
\begin{proof}
See \Cref{sec:proof_l2}.
\end{proof}

\begin{corollary}
  For mean past-discounted stochastic games:
  \begin{enumerate}
    \item the value is determined in mixed stationary strategies over concurrent arenas;
    \item the value is determined in positional strategies over turn-based arenas;
    \item the complexity of the value problem is equivalent to games with mean payoff. 
  \end{enumerate}
\end{corollary}

The results from \cref{sec:future-past}, along with
classical results for discounted and mean-payoff games allow us to extend the Tauberian theorem of \cite{BewleyKohlberg76} and relate asymptotically the values of discounted past-discounted games with the values of mean past-discounted games.

\begin{theorem}[Tauberian Theorem]
	\label{thm:tauberian}
	\begin{equation*}
	    \Val{M P_\gamma, s} = \lim_{\lambda \upto 1} (1 - \gamma) \Val{D_\lambda P_\gamma, s}
	\end{equation*}
\end{theorem}
\begin{proof}
	By \Cref{thm:pd-main}, we obtain the equation
	\begin{equation*}
		\lim_{\lambda \upto 1} (1 - \gamma) \Val{D_\lambda P_\gamma, s} = \lim_{\lambda \upto 1} (1 - \lambda) \frac{\Val{D_\lambda, s}}{1 - \gamma\lambda}.
	\end{equation*}
	Applying \Cref{thm:BewleyKohlberg} to the right-hand side of this equation leads to the following:
	\begin{equation*}
		\lim_{\lambda \upto 1} (1 - \gamma) \Val{D_\lambda P_\gamma, s} = \frac{\Val{M, s}}{1 - \gamma}.
	\end{equation*}
	Finally, applying \Cref{thm:pm-main} yields the desired equivalence:
	\begin{equation*}
		\lim_{\lambda \upto 1} (1 - \gamma) \Val{D_\lambda P_\gamma, s} = \Val{M P_\gamma, s}.
	\end{equation*} 
\end{proof}

\section{Conclusion}
\label{sec:conclusion}
The discounted and mean-payoff objectives have played central roles in the
theory of stochastic games. 
A multitude of deep results exist connecting these objectives
\cite{BewleyKohlberg76,BewleyKohlberg78,MertensNeyman81,AnderssonMiltersen09,ChatterjeeDoyenSingh11,ChatterjeeMajumdar12,Ziliotto16,Ziliotto16G,Ziliotto18}
in addition to an extensive body of work on algorithms for solving these games and their complexity
\cite{FilarSchultz86,RaghavanFilar91,RaghavanSyed03,ChatterjeeMajumdarHenzinger08,EtessamiYannakakis10,ChatterjeeIbsenJensen15}.

Past discounted sums for finite sequences were studied in the context of optimization \cite{alur2012regular} and are closely related to exponential recency
weighted average, a technique used in nonstationary multi-armed bandit
problems~\cite{SuttonBarto98} to estimate the average reward of different
actions by giving more weight to recent outcomes.
However, to the best of our knowledge, past-discounting has not been formally studied as a payoff function for stochastic games until the present work.

Discounted objectives have found significant applications in areas of
program verification and synthesis
\cite{deAlfaroHenzingerMajumdar03,CernyChatterjeeHenzingerRadhakrishnaSing11}. 
Relatively recently---although the idea of past operators is quite old
\cite{LichtensteinPnueliZuck85}---a number of classical formalisms including
temporal logics such as LTL and CTL and the modal $\mu$-calculus have been
extended with past-tense operators and with discounted quantitative semantics
\cite{deAlfaroFaellaHenzingerMajumdarStoelinga05,AlmagorBokerKupferman14,AlmagorBokerKupferman16}.
A particularly significant result \cite{Markey03} around LTL with classical
boolean semantics is that, while LTL with past operators is no more expressive
than standard LTL, it is exponentially more succinct. 
It remains open whether this type of relationship holds for other logics and their extensions by past operators when interpreted with discounted quantitative semantics \cite{AlmagorBokerKupferman16}. 

Regret minimization \cite{Cesa-BianchiLugosi06} is a popular criterion in the setting of online
learning where a decision-maker chooses her actions so as to minimize the average regret---the difference between the realized reward and the reward that
could have been achieved.
We posit that imperfect decision makers may view their regret in a past-discounted
sense, since a suboptimal action in the recent past tends to cause more regret than an equally suboptimal action in the distant past. 
We hope that the results of this work spur further interest in developing
foundations of past-discounted characterizations of regret in online learning
and optimization.

\bibliographystyle{plainnat}
\bibliography{refs}

\begin{thebibliography}{40}
\providecommand{\natexlab}[1]{#1}
\providecommand{\url}[1]{\texttt{#1}}
\expandafter\ifx\csname urlstyle\endcsname\relax
  \providecommand{\doi}[1]{doi: #1}\else
  \providecommand{\doi}{doi: \begingroup \urlstyle{rm}\Url}\fi

\bibitem[Almagor et~al.(2014)Almagor, Boker, and
  Kupferman]{AlmagorBokerKupferman14}
Shaull Almagor, Udi Boker, and Orna Kupferman.
\newblock Discounting in {LTL}.
\newblock In \emph{Tools and Algorithms for the Construction and Analysis of
  Systems, {TACAS}}, volume 8413 of \emph{LNCS}, pages 424--439. Springer,
  2014.
\newblock URL \url{https://doi.org/10.1007/978-3-642-54862-8\_37}.

\bibitem[Almagor et~al.(2016)Almagor, Boker, and
  Kupferman]{AlmagorBokerKupferman16}
Shaull Almagor, Udi Boker, and Orna Kupferman.
\newblock Formally reasoning about quality.
\newblock \emph{J. {ACM}}, 63\penalty0 (3):\penalty0 24:1--24:56, 2016.
\newblock URL \url{https://doi.org/10.1145/2875421}.

\bibitem[Alur et~al.(2012)Alur, D'Antoni, Deshmukh, Raghothaman, and
  Yuan]{alur2012regular}
Rajeev Alur, Loris D'Antoni, Jyotirmoy~V. Deshmukh, Mukund Raghothaman, and
  Yifei Yuan.
\newblock Regular functions, cost register automata, and generalized min-cost
  problems, 2012.

\bibitem[Andersson and Miltersen(2009)]{AnderssonMiltersen09}
Daniel Andersson and Peter~Bro Miltersen.
\newblock The complexity of solving stochastic games on graphs.
\newblock In \emph{Algorithms and Computation {ISAAC}}, volume 5878 of
  \emph{LNCS}, pages 112--121. Springer, 2009.
\newblock URL \url{https://doi.org/10.1007/978-3-642-10631-6\_13}.

\bibitem[Bewley and Kohlberg(1976)]{BewleyKohlberg76}
Truman Bewley and Elon Kohlberg.
\newblock The asymptotic theory of stochastic games.
\newblock \emph{Mathematics of Operations Research}, 1\penalty0 (3):\penalty0
  197--208, 1976.

\bibitem[Bewley and Kohlberg(1978)]{BewleyKohlberg78}
Truman Bewley and Elon Kohlberg.
\newblock On stochastic games with stationary optimal strategies.
\newblock \emph{Mathematics of Operations Research}, 3\penalty0 (2):\penalty0
  104--125, 1978.
\newblock ISSN 0364765X, 15265471.
\newblock URL \url{http://www.jstor.org/stable/3689337}.

\bibitem[Bordais et~al.(2019)Bordais, Guha, and Raskin]{BordaisGuhaRaskin19}
Benjamin Bordais, Shibashis Guha, and Jean{-}Fran{\c{c}}ois Raskin.
\newblock Expected window mean-payoff.
\newblock In \emph{39th {IARCS} Annual Conference on Foundations of Software
  Technology and Theoretical Computer Science, {FSTTCS}}, volume 150 of
  \emph{LIPIcs}, pages 32:1--32:15. Schloss Dagstuhl - Leibniz-Zentrum
  f{\"{u}}r Informatik, 2019.
\newblock URL \url{https://doi.org/10.4230/LIPIcs.FSTTCS.2019.32}.

\bibitem[Cern{\'{y}} et~al.(2011)Cern{\'{y}}, Chatterjee, Henzinger,
  Radhakrishna, and Singh]{CernyChatterjeeHenzingerRadhakrishnaSing11}
Pavol Cern{\'{y}}, Krishnendu Chatterjee, Thomas~A. Henzinger, Arjun
  Radhakrishna, and Rohit Singh.
\newblock Quantitative synthesis for concurrent programs.
\newblock In \emph{Computer Aided Verification {CAV}}, volume 6806 of
  \emph{Lecture Notes in Computer Science}, pages 243--259. Springer, 2011.
\newblock URL \url{https://doi.org/10.1007/978-3-642-22110-1\_20}.

\bibitem[Cesa{-}Bianchi and Lugosi(2006)]{Cesa-BianchiLugosi06}
Nicol{\`{o}} Cesa{-}Bianchi and G{\'{a}}bor Lugosi.
\newblock \emph{Prediction, learning, and games}.
\newblock Cambridge University Press, 2006.
\newblock ISBN 978-0-521-84108-5.
\newblock URL \url{https://doi.org/10.1017/CBO9780511546921}.

\bibitem[Chatterjee and Fijalkow(2011)]{ChatterjeeFijalkow11}
Krishnendu Chatterjee and Nathana{\"{e}}l Fijalkow.
\newblock A reduction from parity games to simple stochastic games.
\newblock In \emph{Proceedings of Second International Symposium on Games,
  Automata, Logics and Formal Verification, GandALF}, volume~54 of
  \emph{{EPTCS}}, pages 74--86, 2011.
\newblock URL \url{https://doi.org/10.4204/EPTCS.54.6}.

\bibitem[Chatterjee and Henzinger(2007)]{ChatterjeeHenzinger07}
Krishnendu Chatterjee and Thomas~A. Henzinger.
\newblock Probabilistic systems with limsup and liminf objectives.
\newblock In \emph{Infinity in Logic and Computation, International Conference,
  {ILC}}, volume 5489 of \emph{Lecture Notes in Computer Science}, pages
  32--45. Springer, 2007.
\newblock URL \url{https://doi.org/10.1007/978-3-642-03092-5\_4}.

\bibitem[Chatterjee and Henzinger(2008)]{ChatterjeeHenzinger08}
Krishnendu Chatterjee and Thomas~A. Henzinger.
\newblock Value iteration.
\newblock In \emph{25 Years of Model Checking - History, Achievements,
  Perspectives}, volume 5000 of \emph{Lecture Notes in Computer Science}, pages
  107--138. Springer, 2008.
\newblock URL \url{https://doi.org/10.1007/978-3-540-69850-0\_7}.

\bibitem[Chatterjee and Ibsen{-}Jensen(2015)]{ChatterjeeIbsenJensen15}
Krishnendu Chatterjee and Rasmus Ibsen{-}Jensen.
\newblock Qualitative analysis of concurrent mean-payoff games.
\newblock \emph{Inf. Comput.}, 242:\penalty0 2--24, 2015.
\newblock URL \url{https://doi.org/10.1016/j.ic.2015.03.009}.

\bibitem[Chatterjee and Majumdar(2012)]{ChatterjeeMajumdar12}
Krishnendu Chatterjee and Rupak Majumdar.
\newblock Discounting and averaging in games across time scales.
\newblock \emph{Int. J. Found. Comput. Sci.}, 23\penalty0 (3):\penalty0
  609--625, 2012.
\newblock URL \url{https://doi.org/10.1142/S0129054112400308}.

\bibitem[Chatterjee et~al.(2008)Chatterjee, Majumdar, and
  Henzinger]{ChatterjeeMajumdarHenzinger08}
Krishnendu Chatterjee, Rupak Majumdar, and Thomas~A. Henzinger.
\newblock Stochastic limit-average games are in {EXPTIME}.
\newblock \emph{Int. J. Game Theory}, 37\penalty0 (2):\penalty0 219--234, 2008.
\newblock URL \url{https://doi.org/10.1007/s00182-007-0110-5}.

\bibitem[Chatterjee et~al.(2009)Chatterjee, Doyen, and
  Henzinger]{ChatterjeeDoyenHenzinger09}
Krishnendu Chatterjee, Laurent Doyen, and Thomas~A. Henzinger.
\newblock A survey of stochastic games with limsup and liminf objectives.
\newblock In \emph{Automata, Languages and Programming, 36th Internatilonal
  Colloquium, {ICALP}}, volume 5556 of \emph{Lecture Notes in Computer
  Science}, pages 1--15. Springer, 2009.
\newblock URL \url{https://doi.org/10.1007/978-3-642-02930-1\_1}.

\bibitem[Chatterjee et~al.(2011)Chatterjee, Doyen, and
  Singh]{ChatterjeeDoyenSingh11}
Krishnendu Chatterjee, Laurent Doyen, and Rohit Singh.
\newblock On memoryless quantitative objectives.
\newblock In \emph{Fundamentals of Computation Theory {FCT}}, volume 6914 of
  \emph{Lecture Notes in Computer Science}, pages 148--159. Springer, 2011.
\newblock URL \url{https://doi.org/10.1007/978-3-642-22953-4\_13}.

\bibitem[Chatterjee et~al.(2013)Chatterjee, Doyen, Randour, and
  Raskin]{ChatterjeeDoyenRandourRaskin13}
Krishnendu Chatterjee, Laurent Doyen, Mickael Randour, and
  Jean{-}Fran{\c{c}}ois Raskin.
\newblock Looking at mean-payoff and total-payoff through windows.
\newblock In \emph{Automated Technology for Verification and Analysis - 11th
  International Symposium, {ATVA}}, volume 8172 of \emph{Lecture Notes in
  Computer Science}, pages 118--132. Springer, 2013.
\newblock URL \url{https://doi.org/10.1007/978-3-319-02444-8\_10}.

\bibitem[Chatterjee et~al.(2015)Chatterjee, Doyen, Randour, and
  Raskin]{ChatterjeeDoyenRandourRaskin15}
Krishnendu Chatterjee, Laurent Doyen, Mickael Randour, and
  Jean{-}Fran{\c{c}}ois Raskin.
\newblock Looking at mean-payoff and total-payoff through windows.
\newblock \emph{Inf. Comput.}, 242:\penalty0 25--52, 2015.
\newblock URL \url{https://doi.org/10.1016/j.ic.2015.03.010}.

\bibitem[de~Alfaro et~al.(2003)de~Alfaro, Henzinger, and
  Majumdar]{deAlfaroHenzingerMajumdar03}
Luca de~Alfaro, Thomas~A. Henzinger, and Rupak Majumdar.
\newblock Discounting the future in systems theory.
\newblock In \emph{Automata, Languages and Programming {ICALP}}, volume 2719 of
  \emph{LNCS}, pages 1022--1037. Springer, 2003.
\newblock URL \url{https://doi.org/10.1007/3-540-45061-0\_79}.

\bibitem[de~Alfaro et~al.(2005)de~Alfaro, Faella, Henzinger, Majumdar, and
  Stoelinga]{deAlfaroFaellaHenzingerMajumdarStoelinga05}
Luca de~Alfaro, Marco Faella, Thomas~A. Henzinger, Rupak Majumdar, and
  Mari{\"{e}}lle Stoelinga.
\newblock Model checking discounted temporal properties.
\newblock \emph{Theor. Comput. Sci.}, 345\penalty0 (1):\penalty0 139--170,
  2005.
\newblock URL \url{https://doi.org/10.1016/j.tcs.2005.07.033}.

\bibitem[Etessami and Yannakakis(2010)]{EtessamiYannakakis10}
Kousha Etessami and Mihalis Yannakakis.
\newblock On the complexity of nash equilibria and other fixed points.
\newblock \emph{{SIAM} J. Comput.}, 39\penalty0 (6):\penalty0 2531--2597, 2010.
\newblock URL \url{https://doi.org/10.1137/080720826}.

\bibitem[Filar and Vrieze(1996)]{FilarVrieze96}
Jerzy Filar and Koos Vrieze.
\newblock \emph{Competitive Markov Decision Processes}.
\newblock Springer-Verlag, Berlin, Heidelberg, 1996.
\newblock ISBN 0387948058.

\bibitem[Filar and Schultz(1986)]{FilarSchultz86}
Jerzy~A. Filar and Todd~A. Schultz.
\newblock Nonlinear programming and stationary strategies in stochastic games.
\newblock \emph{Math. Program.}, 34\penalty0 (2):\penalty0 243--247, 1986.
\newblock URL \url{https://doi.org/10.1007/BF01580590}.

\bibitem[Gawlitza and Seidl(2009)]{GawlitzaSeidl09}
Thomas Gawlitza and Helmut Seidl.
\newblock Games through nested fixpoints.
\newblock In \emph{Computer Aided Verification, 21st International Conference,
  {CAV}}, volume 5643 of \emph{Lecture Notes in Computer Science}, pages
  291--305. Springer, 2009.
\newblock URL \url{https://doi.org/10.1007/978-3-642-02658-4\_24}.

\bibitem[Gimbert(2007)]{Gimbert07}
Hugo Gimbert.
\newblock Pure stationary optimal strategies in markov decision processes.
\newblock In \emph{{STACS}, 24th Annual Symposium on Theoretical Aspects of
  Computer Science}, volume 4393 of \emph{Lecture Notes in Computer Science},
  pages 200--211. Springer, 2007.
\newblock URL \url{https://doi.org/10.1007/978-3-540-70918-3\_18}.

\bibitem[Lichtenstein et~al.(1985)Lichtenstein, Pnueli, and
  Zuck]{LichtensteinPnueliZuck85}
Orna Lichtenstein, Amir Pnueli, and Lenore~D. Zuck.
\newblock The glory of the past.
\newblock In \emph{Logics of Programs, Conference}, volume 193 of \emph{LNCS},
  pages 196--218. Springer, 1985.
\newblock URL \url{https://doi.org/10.1007/3-540-15648-8\_16}.

\bibitem[Liggett and Lippman(1969)]{LiggettLippman69}
T.~M. Liggett and S.~A. Lippman.
\newblock Short notes: Stochastic games with perfect information and time
  average payoff.
\newblock \emph{SIAM Review}, 11\penalty0 (4):\penalty0 604--607, 1969.

\bibitem[Maitra and Sudderth(1992)]{MaitraSudderth92}
A~Maitra and William Sudderth.
\newblock An operator solution of stochastic games.
\newblock \emph{Israel Journal of Mathematics}, 78\penalty0 (1):\penalty0
  33--49, 1992.

\bibitem[Maitra and Sudderth(2012)]{MaitraSudderth12}
Ashok~P Maitra and William~D Sudderth.
\newblock \emph{Discrete gambling and stochastic games}, volume~32.
\newblock Springer Science \& Business Media, 2012.

\bibitem[Markey(2003)]{Markey03}
Nicolas Markey.
\newblock Temporal logic with past is exponentially more succinct, concurrency
  column.
\newblock \emph{Bull. {EATCS}}, 79:\penalty0 122--128, 2003.

\bibitem[Mertens and Neyman(1981)]{MertensNeyman81}
J.F. Mertens and Abraham Neyman.
\newblock Stochastic games.
\newblock \emph{International Journal of Game Theory}, 10\penalty0
  (2):\penalty0 53--66, 1981.

\bibitem[Raghavan and Filar(1991)]{RaghavanFilar91}
T.~E.~S. Raghavan and Jerzy~A. Filar.
\newblock Algorithms for stochastic games - {A} survey.
\newblock \emph{{ZOR} Methods Model. Oper. Res.}, 35\penalty0 (6):\penalty0
  437--472, 1991.
\newblock URL \url{https://doi.org/10.1007/BF01415989}.

\bibitem[Raghavan and Syed(2003)]{RaghavanSyed03}
T.~E.~S. Raghavan and Zamir Syed.
\newblock A policy-improvement type algorithm for solving zero-sum two-person
  stochastic games of perfect information.
\newblock \emph{Math. Program.}, 95\penalty0 (3):\penalty0 513--532, 2003.
\newblock URL \url{https://doi.org/10.1007/s10107-002-0312-3}.

\bibitem[Rudin(1964)]{Rudin64}
Walter Rudin.
\newblock \emph{Principles of mathematical analysis}, volume~3.
\newblock McGraw-hill New York, 1964.

\bibitem[Shapley(1953)]{Shapley53}
L.~S. Shapley.
\newblock Stochastic games.
\newblock \emph{Proceedings of the National Academy of Sciences}, 39\penalty0
  (10):\penalty0 1095--1100, 1953.
\newblock ISSN 0027-8424.
\newblock URL \url{https://www.pnas.org/content/39/10/1095}.

\bibitem[Sutton and Barto(1998)]{SuttonBarto98}
Richard~S. Sutton and Andrew~G. Barto.
\newblock \emph{Reinforcement learning - an introduction}.
\newblock Adaptive computation and machine learning. {MIT} Press, 1998.
\newblock ISBN 978-0-262-19398-6.
\newblock URL \url{https://www.worldcat.org/oclc/37293240}.

\bibitem[Ziliotto(2016{\natexlab{a}})]{Ziliotto16}
Bruno Ziliotto.
\newblock A tauberian theorem for nonexpansive operators and applications to
  zero-sum stochastic games.
\newblock \emph{Mathematics of Operations Research}, 41\penalty0 (4):\penalty0
  1522--1534, 2016{\natexlab{a}}.
\newblock URL \url{https://doi.org/10.1287/moor.2016.0788}.

\bibitem[Ziliotto(2016{\natexlab{b}})]{Ziliotto16G}
Bruno Ziliotto.
\newblock General limit value in zero-sum stochastic games.
\newblock \emph{Int. J. Game Theory}, 45\penalty0 (1-2):\penalty0 353--374,
  2016{\natexlab{b}}.
\newblock URL \url{https://doi.org/10.1007/s00182-015-0509-3}.

\bibitem[Ziliotto(2018)]{Ziliotto18}
Bruno Ziliotto.
\newblock Tauberian theorems for general iterations of operators: Applications
  to zero-sum stochastic games.
\newblock \emph{Games Econ. Behav.}, 108:\penalty0 486--503, 2018.
\newblock URL \url{https://doi.org/10.1016/j.geb.2018.01.009}.

\end{thebibliography}

\appendix

\section{Additional Details from the Proof of Theorem 5}
\label{sec:window_addendum}
Define $h : \Sigma^G_\Min \cup \Sigma_\Max^G \to \Sigma^H_\Min \cup \Sigma_\Max^H$ and $g : \Sigma^H_\Min \cup \Sigma_\Max^H \to \Sigma^G_\Min \cup \Sigma_\Max^G$ as the functions satisfying the following equations, for $\sigma \in \Sigma^G_\Min \cup \Sigma^G_\Max$ and $\varsigma \in \Sigma^H_\Min \cup \Sigma^H_\Max$, corresponding to \cref{eq:strategy_bijection}.
\begin{align*}
    \sigma(\subseq{\pi}{0}{n}, s) &= h(\sigma)\left(\subseq{\varpi}{0}{n}, \left(s, \subseq{\pi}{\max(0, n-1-\ell)}{n-1}\right)\right) \\
    g(\varsigma)(\subseq{\pi}{0}{n}, s) &= \varsigma\left(\subseq{\varpi}{0}{n}, \left(s, \subseq{\pi}{\max(0, n-1-\ell)}{n-1}\right)\right)
\end{align*}
Let $h(\Sigma^G_\Min), h(\Sigma^G_\Max)$ be the the image of $h$ on all strategies for the players over $G$, and let $g(\Sigma^H_\Min), g(\Sigma^H_\Max)$ be the image of $g$ on all strategies of the players over $H$.

For every $\nu \in \Sigma_\Min^G$ there exists $h(\nu) \in \Sigma_\Min^H$ and for every $\chi \in \Sigma_\Max^G$ there exists $h(\chi) \in \Sigma_\Max^H$ such that:
\begin{align}
    \bb{E}^{\nu,\chi}_s(W_\ell P_\gamma) &=  \bb{E}^{h(\nu),h(\chi)}_{(s,q_0)}(L) \\
    \sup_{\chi\in \Sigma_\Max^G}  \bb{E}^{\nu,\chi}_s(W_\ell P_\gamma) &=  \sup_{\chi \in h(\Sigma_\Max^G)}  \bb{E}^{h(\nu),\chi}_{(s,q_0)}(L) \nonumber \\
    \inf_{\nu \in \Sigma_\Min^G} \sup_{\chi \in \Sigma_\Max^G}  \bb{E}^{\nu,\chi}_s(W_\ell P_\gamma) &= \inf_{\nu \in h(\Sigma^G_\Min)} \sup_{\chi \in h(\Sigma^G_\Max)}  \bb{E}^{\nu,\chi}_{(s,q_0)}(L) \\
    \inf_{\nu \in \Sigma^G_\Min} \sup_{\chi \in \Sigma^G_\Max}  \bb{E}^{\nu,\chi}_s(W_\ell P_\gamma) &= \inf_{\nu \in \Sigma^H_\Min} \sup_{\chi \in \Sigma^H_\Max}  \bb{E}^{\nu,\chi}_{(s,q_0)}(L).
\end{align}
Symmetrically, we have
\begin{align}
    \bb{E}^{\nu,\chi}_s(W_\ell P_\gamma) &=  \bb{E}^{h(\nu),h(\chi)}_{(s,q_0)}(L) \\
    \inf_{\nu \in \Sigma_\Min^G}  \bb{E}^{\nu,\chi}_s(W_\ell P_\gamma) &=  \inf_{\nu \in h(\Sigma_\Max^G)}  \bb{E}^{\nu,h(\chi)}_{(s,q_0)}(L) \nonumber \\
    \sup_{\chi \in \Sigma_\Max^G} \inf_{\nu \in \Sigma_\Min^G} \bb{E}^{\nu,\chi}_s(W_\ell P_\gamma) &=  \sup_{\chi \in h(\Sigma^G_\Max)} \inf_{\nu \in h(\Sigma^G_\Min)} \bb{E}^{\nu,\chi}_{(s,q_0)}(L) \\
    \sup_{\chi \in \Sigma^G_\Max} \inf_{\nu \in \Sigma^G_\Min} \bb{E}^{\nu,\chi}_s(W_\ell P_\gamma) &= \sup_{\chi \in \Sigma^H_\Max} \inf_{\nu \in \Sigma^H_\Min} \bb{E}^{\nu,\chi}_{(s,q_0)}(L).
\end{align}

Lines (12) and (15) follow from the arguments given in the main body of the proof of \Cref{thm:window_liminf}.
Because of the fact that $h$ is bijective, we have that $h(\Sigma^G_\Min) = \Sigma^H_\Min$,  for \Min{}, and likewise for \Max{}, mutatis mutandis.
This provides justification for obtaining (14) from (13) and for obtaining (17) from (16) in the above derivations.

For every $\nu \in \Sigma_\Min^H$ there exist $g(\nu) \in \Sigma_\Min^G$ and for every $\chi \in \Sigma_\Max^H$ there exists an $g(\chi) \in \Sigma_\Max^G$ such that:
\begin{align}
    \bb{E}^{g(\nu),g(\chi)}_s(W_\ell P_\gamma) &=  \bb{E}^{\nu,\chi}_{(s,q_0)}(L) \\
    \sup_{\chi \in g(\Sigma_\Max^H)}  \bb{E}^{g(\nu),\chi}_s(W_\ell P_\gamma) &=  \sup_{\chi \in \Sigma_\Max^H}  \bb{E}^{\nu,\chi}_{(s,q_0)}(L) \nonumber \\
    \inf_{\nu \in g(\Sigma_\Min^H)} \sup_{\chi \in g(\Sigma_\Max^H)}  \bb{E}^{\nu,\chi}_s(W_\ell P_\gamma) &= \inf_{\nu \in \Sigma^H_\Min} \sup_{\chi \in \Sigma^H_\Max} \bb{E}^{\nu,\chi}_{(s,q_0)}(L) \\
    \inf_{\nu \in \Sigma^G_\Min} \sup_{\chi \in \Sigma^G_\Max}  \bb{E}^{\nu,\chi}_s(W_\ell P_\gamma) &= \inf_{\nu \in \Sigma^H_\Min} \sup_{\chi \in \Sigma^H_\Max}  \bb{E}^{\nu,\chi}_{(s,q_0)}(L).
\end{align}
Symmetrically, we have
\begin{align}
    \bb{E}^{g(\nu),g(\chi)}_s(W_\ell P_\gamma) &=  \bb{E}^{\nu,\chi}_{(s,q_0)}(L) \\
    \inf_{\nu \in g(\Sigma_\Min^H)}  \bb{E}^{\nu,g(\chi)}_s(W_\ell P_\gamma) &= \inf_{\nu \in \Sigma_\Min^H}  \bb{E}^{\nu,\chi}_{(s,q_0)}(L) \nonumber \\
    \sup_{\chi \in g(\Sigma_\Max^H)} \inf_{\nu \in g(\Sigma_\Min^H)} \bb{E}^{\nu,\chi}_s(W_\ell P_\gamma) &= \sup_{\chi \in \Sigma^H_\Max} \inf_{\nu \in \Sigma^H_\Min} \bb{E}^{\nu,\chi}_{(s,q_0)}(L) \\
    \sup_{\chi \in \Sigma^G_\Max} \inf_{\nu \in \Sigma^G_\Min} \bb{E}^{\nu,\chi}_s(W_\ell P_\gamma) &= \sup_{\chi \in \Sigma^H_\Max} \inf_{\nu \in \Sigma^H_\Min}  \bb{E}^{\nu,\chi}_{(s,q_0)}(L).
\end{align}

Lines (18) and (21) follow from the arguments given in the main body of the proof of \Cref{thm:window_liminf}.
Because of the fact that $g$ is bijective, we have that $\Sigma^G_\Min = g(\Sigma^H_\Min)$, and likewise for \Max{}, mutatis mutandis.
This provides justification for obtaining (22) from (23) and for obtaining (27) from (26) in the above derivations.

\section{Proof of Lemma 1}
\label{sec:proof_l1}
For ease of notation, let $w_k = w(s_k, a_k, b_k)$.

\paragraph*{Base case:}
\textit{\Cref{lem:helper1} holds for $n = 0$.}

Let $n = 0$, then 
\begin{equation*}
    \frac{1}{n+1} \sum_{k=0}^{n} \left( \sum_{i=0}^{k} \gamma^{k-i} w_i  \right) = w_0 = \sum^n_{k=0} \frac{w_k (1 - \gamma^{n+1-k})}{(n+1) (1-\gamma)}.
\end{equation*}

\paragraph*{Inductive case:}
\textit{If \Cref{lem:helper1} holds for $n-1$, then it also holds for $n$.}

Firstly, observe the following derivation:
\begin{align*}
    \frac{1}{n+1} \sum_{k=0}^{n} \left( \sum_{i=0}^{k} \gamma^{k-i} w_i \right) &= \sum^n_{k=0} \frac{1}{n+1} \sum^k_{i=0} \gamma^{k-i} w_i \\
    &= \left( \sum^{n-1}_{k=0} \frac{1}{n+1} \sum^k_{i=0} \gamma^{k-i} w_k \right) + \frac{1}{n+1} \sum^n_{k=0} \gamma^{n-k} w_k \\
    &= \frac{n}{n+1} \left( \sum^{n-1}_{k=0} \frac{1}{n} \sum^k_{i=0} \gamma^{k-i} w_k \right) + \frac{1}{n+1} \sum^n_{k=0} \gamma^{n-k} w_k
\end{align*}
Notice that the expression within the parentheses matches exactly the left-hand side of \Cref{lem:helper1} for $n-1$.
By the inductive hypothesis, we may rewrite this as follows:
\begin{equation*}
    \frac{1}{n+1} \sum_{k=0}^{n} \left( \sum_{i=0}^{k} \gamma^{k-i} w_i \right) = \left( \frac{n}{n+1} \sum^{n-1}_{k=0} \frac{w_k (1 - \gamma^{n-k})}{n (1-\gamma)} \right) + \left( \frac{1}{n+1} \sum^n_{k=0} w_k \gamma^{n-k} \right).
\end{equation*}
On the right-hand side, the $n$ terms in the left summand cancel, and factoring $(1-\gamma)$ out of the denominator leaves
\begin{equation*}
    \left( \frac{1}{1-\gamma} \sum^{n-1}_{k=0} \frac{w_k (1 - \gamma^{n-k})}{n+1} \right) + \left( \sum^n_{k=0} \frac{w_k \gamma^{n-k}}{n+1} \right).
\end{equation*}
Now, factoring again by $\frac{1}{1-\gamma}$ and by $\frac{1}{n+1}$ yields
\begin{equation*}
    \frac{\left( \sum\limits^{n-1}_{k=0} w_k (1 - \gamma^{n-k}) \right) + \left( \sum\limits^n_{k=0} w_k \gamma^{n-k} (1-\gamma) \right)}{(n+1) (1-\gamma)}.
\end{equation*}
Multiplying the terms within the summations in the numerator, we get the following expression:
\begin{equation*}
    \frac{\left( \sum\limits^{n-1}_{k=0} w_k - w_k\gamma^{n-k} \right) + \left( \sum\limits^n_{k=0}  w_k\gamma^{n-k} - w_k\gamma^{n+1-k} \right)}{(n+1) (1-\gamma)}.
\end{equation*}
Taking advantage of additive cancelation in the numerator, we obtain
\begin{equation*}
    \frac{\left( \sum\limits^{n-1}_{k=0} w_k \right) + w_n - \left( \sum\limits^n_{k=0} w_k\gamma^{n+1-k} \right)}{(n+1) (1-\gamma)}.
\end{equation*}
With some simple regrouping of terms, the above may be rewritten as
\begin{equation*}
    \frac{\sum\limits^n_{k=0} w_k - \sum\limits^n_{k=0} w_k \gamma^{n+1-k}}{(n+1) (1-\gamma)},
\end{equation*}
and subsequently
\begin{equation*}
    \frac{\sum\limits^n_{k=0} w_k (1 - \gamma^{n+1-k})}{(n+1) (1-\gamma)}.
\end{equation*}
Finally, by the linearity of summation, we obtain the desired equality:
\begin{equation*}
    \frac{1}{n+1} \sum_{k=0}^{n} \left( \sum_{i=0}^{k} \gamma^{k-i} w_i  \right) = \sum^n_{k=0} \frac{w_k (1 - \gamma^{n+1-k})}{(n+1) (1-\gamma)}.
\end{equation*}\qed

\section{Proof of Lemma 2}
\label{sec:proof_l2}
By assumption the sequence $w(\pi) = \seq{w(s_n, a_n, b_n)}_{n \geq 0} = \seq{w_n}_{n \geq 0}$ is bounded, and so the values $\overline{w} = \sup w(\pi)$ and $\underline{w} = \inf w(\pi)$ are well-defined as greatest and least scalar values occurring anywhere in $w(\pi)$.
The following derivation shows that $\sum^n_{k=0} \frac{w_k \gamma^{n+1-k}}{(n+1) (1-\gamma)}$ is bounded above by zero.
\begin{align*}
    \lim_{n \to \infty} \sum^n_{k=0} \frac{\overline{w} \gamma^{n+1-k}}{(n+1) (1-\gamma)} &= \frac{1}{1 - \gamma} \lim_{n \to \infty} \left( \frac{1}{(n+1) (1-\gamma)} \sum^n_{k=0} \overline{w} \gamma^{n+1-k} \right) \\
    &= \frac{1}{1 - \gamma} \lim_{n \to \infty} \left( \left( \frac{1}{n+1} \right) \left( \frac{\overline{w} (1 - \gamma^{n+1})}{1 - \gamma} \right) \right) \\
    &= \frac{1}{1 - \gamma} \lim_{n \to \infty} \left( \frac{1}{n+1} \right) \lim_{n \to \infty} \left( \frac{\overline{w} (1 - \gamma^{n+1})}{1 - \gamma} \right) \\
    &= \frac{1}{1 - \gamma} \cdot 0 \cdot \frac{\overline{w}}{1 - \gamma} = 0
\end{align*}
Similarly, we may obtain an identical equation using $\inf w(\pi)$:
\begin{equation*}
    \lim_{n \to \infty} \sum^n_{k=0} \frac{\underline{w} \gamma^{n+1-k}}{(n+1) (1-\gamma)} = 0
\end{equation*}
showing that $\sum^n_{k=0} \frac{w_k \gamma^{n+1-k}}{(n+1) (1-\gamma)}$ is bounded below by zero as well.
Since this sequence is bounded on both sides by zero, it must converge to zero. \qed

\end{document}